\documentclass{article}
\usepackage{natbib}
\usepackage{fullpage}

\usepackage{amsmath,amssymb,amsthm}
\usepackage{rotating,stmaryrd}
\usepackage[all]{xy}
\usepackage{url}

\newcommand{\To}{\Rightarrow} 
\newcommand{\isoto}{\stackrel{\simeq}{\rightarrow}} 
\newcommand{\rpto}{\rightsquigarrow}

\newcommand{\seff}{\approx} 
\newcommand{\cons}{\ensuremath{\lhd}}
\newcommand{\rcons}{\ensuremath{\rhd}}

\newcommand{\lrcons}{\cons\rcons} 
\newcommand{\consup}{\begin{turn}{-90}\ensuremath{\cons}\end{turn}}
\newcommand{\consdown}{\begin{turn}{90}\ensuremath{\cons}\end{turn}}
\newcommand{\Cons}{\ensuremath{\blacktriangleleft}}
\newcommand{\rCons}{\ensuremath{\blacktriangleright}} 
\newcommand{\lrCons}{\Cons\rCons} 
\newcommand{\Consup}{\begin{turn}{-90}\ensuremath{\Cons}\end{turn}}
\newcommand{\Consdown}{\begin{turn}{90}\ensuremath{\Cons}\end{turn}}

\newcommand{\stimes}{\!\times\!} 
\newcommand{\ltimessp}{\ltimes}
\newcommand{\rtimessp}{\rtimes}
\newcommand{\ltimesseq}{\ltimes_{\mathrm{seq}}} 
\newcommand{\rtimesseq}{\rtimes_{\mathrm{seq}}} 
\newcommand{\ltimesfr}{\ltimes_{\mathrm{Fr}}}
\newcommand{\rtimesfr}{\rtimes_{\mathrm{Fr}}}
\newcommand{\ltimeskl}{\ltimes_{\mathrm{Kl}}}

\newcommand{\arr}{\mathtt{arr}}
\newcommand{\first}{\mathtt{first}}
\newcommand{\scond}{\mathtt{second}}
\newcommand{\acomp}{>\!\!>\!\!>}
\newcommand{\altimes}{*\!\!*\!\!*}
\newcommand{\atuple}{\&\!\!\&\!\!\&}
\newcommand{\assoc}{\mathtt{assoc}}
\newcommand{\swap}{\mathtt{swap}}
\newcommand{\fst}{\mathtt{fst}}

\newcommand{\tA}{\mathtt{A}}
\newcommand{\boxEff}{\makebox[130pt]{Cartesian effect categories}}
\newcommand{\boxArr}{\makebox[150pt]{Arrows}}

\newcommand{\tuple}[1]{\langle #1 \rangle} 
\newcommand{\tu}{\langle \, \rangle} 
\newcommand{\cotu}{!} 
\newcommand{\uno}{1} 
\newcommand{\zero}{0}
\newcommand{\cD}[1]{\mathcal{D}_{#1}} 
\newcommand{\olcD}[1]{\overline{\mathcal{D}}_{#1}} 
\newcommand{\cE}{\mathcal{E}} 
\newcommand{\kl}[1]{[#1]}  
\newcommand{\lk}[1]{]#1[}  
\newcommand{\cL}{\mathcal{L}} 
\newcommand{\len}[1]{\mathrm{len}(#1)} 
\newcommand{\cM}{\mathcal{M}_{\mathrm{fin}}} 
\newcommand{\cP}{\mathcal{P}_{\mathrm{fin}}} 
\newcommand{\card}[1]{\mathrm{card}(#1)} 
\newcommand{\res}{\Downarrow}

\newcommand{\ul}[1]{\underline{#1}}
\newcommand{\id}{\mathrm{id}} 
\newcommand{\et}{\;\mbox{ and }\;} 
\def\N{\mathbb N}
\newcommand{\mysep}{\,,\,}  
\newcommand{\ff}{\varphi}
\renewcommand{\gg}{\psi}
\newcommand{\lift}{\dagger}
\newcommand{\ulx}{\underline{x}}
\newcommand{\ppt}[1]{\textit{#1}}
\newcommand{\formula}[1]{#1}
\newcommand{\slogan}[1]{\textsl{#1}}
\newcommand{\exam}[1]{\null \noindent \textit{#1}.}

\renewcommand{\arraystretch}{1.3} 

\theoremstyle{plain}
\newtheorem{thm}{Theorem}
\newtheorem{lem}[thm]{Lemma}
\newtheorem{prop}[thm]{Proposition}

\theoremstyle{definition}
\newtheorem{defn}{Definition}[section]
\theoremstyle{remark}
\newtheorem{rem}{Remark}

\title{Cartesian effect categories are Freyd-categories}
\author{
Jean-Guillaume Dumas 
\thanks{Laboratoire Jean Kuntzmann, Universit\'e de Grenoble, 38041 Grenoble, France 
-- Jean-Guillaume.Dumas@imag.fr -- http://ljk.imag.fr/membres/Jean-Guillaume.Dumas} 
\and Dominique Duval
\thanks{Laboratoire Jean Kuntzmann, Universit\'e de Grenoble, 38041 Grenoble, France 
-- Dominique.Duval@imag.fr -- http://ljk.imag.fr/membres/Dominique.Duval}
\and Jean-Claude Reynaud
\thanks{Malhivert, 38640 Claix, France -- Jean-Claude.Reynaud@imag.fr}
}

\date{June 12., 2009}

\begin{document}

\maketitle 

\begin{abstract}
Most often, in a categorical semantics for a programming language, 
the substitution of terms is expressed by composition and finite products.
However this does not deal with the order of evaluation of arguments,
which may have major consequences when there are side-effects. 
In this paper Cartesian effect categories are introduced for 
solving this issue, and they are compared with 
strong monads, Freyd-categories and Haskell's Arrows.
It is proved that a Cartesian effect category is a Freyd-category 
where the premonoidal structure is provided by a kind of 
binary product, called the sequential product.
The universal property of the sequential product 
provides Cartesian effect categories with a powerful tool for constructions and proofs. 
To our knowledge, both 
effect categories and sequential products are new notions. 

\textbf{Keywords.}
Categorical logic, computational effects, 
monads, Freyd-categories, premonoidal categories, Arrows, 
sequential product, effect categories, Cartesian effect categories.
\end{abstract}

\section{Introduction}


A categorical semantics for a programming language usually 
associates an object to each type, a morphism to each term, 
and uses composition and finite products for dealing with 
the substitution of terms.
This framework behaves very well in a simple equational setting, 
but it has to be adapted as soon as there is some kind of computational effects, 
for instance non-termination or state updating in an imperative language. 
Then there are two kinds of terms: the general terms may cause effects
while the \emph{pure} terms are effect-free.
Following \citep{Moggi91}, a general term may be seen as a \emph{program} 
that returns a \emph{value} which is pure. 
In this paper we focus on the following \emph{sequentiality} issue:
the categorical products do not deal with the order of evaluation of the arguments,
although this order may have major consequences when there are side-effects. 
For solving this sequentiality issue, we introduce \emph{Cartesian effect categories} 
as an alternative for Cartesian categories.

Other approaches include 
strong monads \citep{Moggi89}, Freyd-categories \citep{PowerRobinson97} 
and Arrows \citep{Hughes00}.
These frameworks are quite similar from several points of view 
\citep{HeunenJacobs06,Atkey08},  
while our framework is more precise. 
A first draft for Cartesian effect categories can be found in \citep{DumDuvRey07},
and a similar approach in \citep{DuvRey05}.


A category is called Cartesian if it has finite products,
and a subcategory $C$ of a category $K$ 
is called wide if it has the same objects as $K$. 
A \emph{Freyd-category} is a generalization of a Cartesian category 
that consists essentially in a category $K$ 
with a wide subcategory $C$, 
such that $C$ is Cartesian (hence $C$ is symmetric monoidal) 
and $K$ is symmetric premonoidal.
A \emph{Cartesian effect category}, as defined in this paper,  
is more precise and more homogeneous than a Freyd-category:
like the symmetric monoidal structure on $C$ derives from its product, 
in a Cartesian effect category the symmetric premonoidal structure on $K$
derives from some kind of product, called a \emph{sequential product},
which extends the product of $C$
and generalizes the usual categorical product.
In fact, there are two steps in our definition.
First an \emph{effect category} is defined, 
without mentioning any kind of product:
it is made of a category $K$ with a wide subcategory $C$
and with a relation $\cons$ called \emph{consistency} between morphisms. 
Then a  \emph{Cartesian effect category} is defined
as an effect category with a binary product on $C$ extended by 
a \emph{sequential product} on $K$, 
which itself is defined thanks to 
a universal property that generalizes the categorical product property 
and involves the consistency relation.
Like every universal property, this 
provides a powerful tool for constructions and proofs
in a Cartesian effect category. 


Let us look at two basic examples of effect categories
(two morphisms in a category are called parallel if they 
share the same domain and the same codomain).

The \emph{non-termination} effect involves partial functions.
As usual, two partial functions are called \emph{consistent} 
when they coincide on the intersection of their domains of definition. 
Thus, on the one hand, 
two partial functions $f$ and $f'$ are consistent if and only if
there is a total 
function $v$ such that $v$ is consistent both with $f$ and with $f'$. 
On the other hand, let us say that two partial functions have the
\emph{same effect}
if they have the same domain of definition.
Then clearly,
two partial functions have the same effect and are consistent if and only if 
they are equal.

In an imperative programming language, 
there are side-effects due to the modification of the \emph{state}, 
since the functions in the sense of the programming language,  
in addition to have arguments and a return value, 
are allowed to use the state and to modify it.
A function is called \emph{pure} if it neither use nor modify the state,
and the side-effects are due to the non-pure functions.
Let us say that a function $f$ is \emph{consistent} with a pure function $v$ 
when both return the same value when they are given the same arguments.
Then two arbitrary functions are called \emph{consistent} when 
they are consistent with a common pure function,
which means that both return the same value when they are given the same arguments 
and that in addition this value does not depend on the state.
It should be noted that this consistency relation is not reflexive.
Therefore, if two functions have the same effect and are consistent 
then they are equal,
but the converse is false.


More generally,
an \emph{effect category} is a category $K$ with a wide subcategory $C$ 
and with a consistency relation $\cons$ between parallel morphisms,
the first one in $K$ and the second one in $C$,
satisfying a form of compatibility with the composition. 
The morphisms in $C$ are called \emph{pure} and are denoted with $\rpto$. 
Two morphisms in $K$ are called \emph{consistent} when there is a pure morphism $v$
such that $f\cons v $ and $f'\cons v $;
this is denoted $f\lrcons f'$, and the properties of consistency
are such that the relation $\lrcons$ extends $\cons$.
Let $\uno$ be a terminal object in $C$,
the \emph{effect} of a morphism $f$ is defined 
as the morphism $\cE(f)=\tu_Y\circ f$ where $\tu_Y$
is the unique pure morphism $\tu_Y:Y\rpto \uno$.
It is assumed that the following \emph{complementarity} property holds,
which means that the consistency relation is a kind of ``up-to-effects'' relation:
\slogan{if two morphisms have the same effect and are consistent, 
then they are equal.}

This notion of consistency coincides with the usual one for partial functions,
but to our knowledge it is new in the general setting of computational effects. 
For instance, we will see in section~\ref{subsec:result} that it is fairly different 
from the notion of \emph{having the same result} that  is defined
in \citep{Moggi95} in the framework of evaluation logic. 
Let us look more closely at the complementarity property 
(for some fixed domain and codomain).
On the one hand, to have the same effect is an equivalence relation $\seff$ 
with one distinguished equivalence class, the class 
of the morphisms without effect, which contains all the pure morphisms.
On the other hand, to be consistent is a symmetric relation $\lrcons$, 
with each maximal clique made of a unique pure morphism and all 
the morphisms that are consistent with it. 
The complementarity property asserts that 
there is at most one morphism in the intersection
of a given equivalence class for $\seff$
and a given maximal clique for $\lrcons$. 


A binary product on a category $C$ provides a bifunctor $\stimes$ 
on $C$ such that for all $v_1:X_1\to Y_1$ and $v_2:X_2\to Y_2$, 
the morphism $v_1 \stimes v_2:X_1\stimes X_2 \to Y_1\stimes Y_2$ 
is characterized by the following diagram, 
where the $p_i$'s and $q_i$'s are the projections.
This property is symmetric in $v_1$ and $v_2$.
When $C$ is the category of sets, this means that 
$(v_1 \stimes v_2)(x_1,x_2)=\tuple{v_1(x_1),v_2(x_2)}$.

  $$
  \xymatrix@R=1pc@C=3pc{
  X_1 \ar[r]^{v_1} \ar@{}[rd]|{=} & Y_1 \\
  X_1\times X_2 \ar[r]^{v_1\times v_2} \ar[u]^{p_1} \ar[d]_{p_2} & 
  Y_1\times Y_2 \ar[u]_{q_1} \ar[d]^{q_2} \\
  X_2 \ar[r]^{v_2} & Y_2 \ar@{}[lu]|{=} \\
  }$$


A  \emph{Cartesian effect category} is defined
as an effect category with a binary product on $C$, extended by two symmetric 
\emph{semi-pure products} $v\ltimes f$ and $f\rtimes v$ where $v$ is pure.
The left semi-pure product $v \ltimes f$ is characterized by the following diagram, 
which means that 
$q_1\circ(v \ltimes f)\cons v\circ p_1$ and $q_2\circ(v \ltimes f)=f\circ p_2$ 
(the right semi-pure product is characterized by a symmetric diagram).

  $$
  \xymatrix@R=1pc@C=3pc{
  Y_1 \ar@{~>}[r]^{v} \ar@{}[rd]|{\consdown} & Y_1 \\
  Y_1\times X_2 \ar[r]^{v \ltimes f} \ar@{~>}[u]^{p_1} \ar@{~>}[d]_{p_2} & 
  Y_1\times Y_2 \ar@{~>}[u]_{q_1} \ar@{~>}[d]^{q_2} \\
  X_2 \ar[r]^{f} & Y_2 \ar@{}[lu]|{=} \\
  }$$

This property means that the effect of $v\ltimes f$ is the effect of $f$, 
and that ``up to effects'' $v\ltimes f$ looks like an ordinary binary product. 
Then the left \emph{sequential product} of two arbitrary morphisms
$f_1$ and $f_2$ is easily obtained by composing two semi-pure products: 
$f_1 \ltimes f_2 = (\id_1 \ltimes f_2) \circ (f_1 \rtimes \id_2)$
where $\id_1$ and $\id_2$ denote the identities of $Y_1$ and $X_2$, respectively.
This definition formalizes the notion of \emph{sequentiality}:
``first $f_1$, then $f_2$''. The right  sequential product is defined in a symmetric way. 
We will check that the sequential product extends the semi-pure product, 
so that there is no ambiguity in using the same symbols $\ltimes$ and $\rtimes$ for both.
This approach, to our knowledge, is completely new.
It can be summarized as follows: 
\slogan{while the universal property of a binary product consists in two equalities, 
the universal property of a semi-pure product consists in one equality 
and one consistency.}

For instance, in the category of sets with partial functions, 
$v \ltimes f$ is the partial function such that 
$(v \ltimes f)(x_1,x_2)=\tuple{y_1,y_2}$ 
where $y_1=v(x_1)$ and $y_2=f(x_2)$ whenever $f(x_2)$ is defined,
otherwise $(v \ltimes f)(x_1,x_2)$ is not defined. 
When side-effects are due to the updating of the state,
$v \ltimes f$ is such that for each state $s$, 
$(v \ltimes f)(s,x_1,x_2)=\tuple{s_2,y_1,y_2}$ 
where $\tuple{s,y_1}=v(s,x_1)$ and $\tuple{s_2,y_2}=f(s,x_2)$.

The properties of the sequential product imply 
that a Cartesian effect category is a Freyd-category. 
On the other hand, 
each strong monad defines a Freyd-category \citep{PowerRobinson97}. 
We prove that a Freyd-category defined from a strong monad 
is a weak Cartesian effect category if and only if, roughly speaking: 
\slogan{the strength of the monad is consistent with the identity.} 


Section~\ref{sec:effect} is devoted to effect categories and  
section~\ref{sec:cartesian} to Cartesian effect categories.
Then Cartesian effect categories are related to 
Freyd-categories, Arrows and strong monads in section~\ref{sec:compare}. 
Several examples are considered in sections~\ref{subsec:effect-exam},
\ref{subsec:cartesian-exam} and~\ref{subsec:monad-exam}.

\section{Effect categories}
\label{sec:effect}

\subsection{Pure morphisms}
\label{subsec:pure}

\begin{defn}
\label{defn:pure}
A subcategory $C$ of a category $K$ is \emph{wide} if it has the same objects as $K$;
this is denoted $C\subseteqq K$.
Given $C\subseteqq K$, 
a morphism of $K$ is called \emph{pure} if it is in $C$;
then it is denoted with ``$\rpto$''.
An object $\uno$ is a \emph{pure terminal} object in $C\subseteqq K$ if it is terminal in $C$,
then for each object $X$ the unique pure morphism from $X$ to $\uno$ 
is denoted $\tu_X:X\rpto\uno$. 
\end{defn}

\begin{rem}\textbf{Pure morphisms in a Kleisli category.} 
\label{rem:kl-pure}
Let $C_0$ be a category (called the base category) 
with a monad $(M,\mu,\eta)$ (or simply $M$) 
and let $K_M$ be the Kleisli category of $M$. 
Then $K_M$ has the same objects as $C_0$ and 
for all objects $X$ and $Y$ there is a bijection between $C_0(X,MY)$ and $K_M(X,Y)$.
In this paper, 
for each morphism $f:X\to Y$ in $K_M$ 
the corresponding morphism in $C_0$ is denoted $\kl{f}:X\to MY$, 
and we say that $f$ \emph{stands for} $\kl{f}$,
and for each morphism $\ff:X\to MY$ in $C_0$ 
the corresponding morphism in $K_M$ is denoted $\lk{\ff}:X\to Y$. 
So, $\lk{(\kl{f})}=f$ for every $f$ in $K_M$
and $\kl{(\lk{\ff})}=\ff$ for every $\ff$ in $C_0$ with codomain $MY$ for some $Y$.
Let $J:C_0\to K_M$ denote the functor associated with $M$
and let $C_M=J(C_0)$. Then $J$ is the identity on objects, 
so that $C_M$ is a wide subcategory of $K_M$.
A pure morphism $v:X\rpto Y$ in $K_M$ is a morphism $v=J(v_0)$
for some $v_0:X\to Y$ in $C_0$; this  means that 
$\kl{v}=\eta_Y\circ v_0:X\to MY$ in $C_0$.
Each identity $\id_X$ in $K_M$ henceforth stands for $\kl{\id_X}=\eta_X$ 
and the composition $g\circ f$ of $f:X\to Y$ and $g:Y\to Z$ 
stands for $\kl{g\circ f} = \kl{g}^* \circ \kl{f}$ where $\kl{g}^*=\mu_Z\circ M\kl{g}$.
It follows that when $v:X\rpto Y$ and $w:Y\rpto Z$, then 
$\kl{g\circ v} = \kl{g} \circ v_0$,
$\kl{w\circ f} =  Mw_0 \circ \kl{f}$
and $\kl{w\circ v} = \eta_Z\circ w_0 \circ v_0$.
It should be noted that it does not make sense to say that 
a morphism in $C_0$ is pure or not. 
Indeed, each morphism $\ff:X\to MY$ in $C_0$ gives rise in $K_M$ 
both to a pure morphism $v=J(\ff):X\to MY$
and to a morphism $f=\lk{\ff}:X\to Y$,
related by $\kl{v} = \eta_{MY} \circ \kl{f}$ in $C_0$.
$$ \begin{array}{|c|c|c|c|}
\hline
C_0 &  
\xymatrix@C=3pc{X \ar[r]^{\kl{f}} & MY \\} &
\xymatrix@C=3pc@R=1pc{X \ar[r]^{\kl{v}} \ar[rd]_{v_0} & MY \\ 
   \ar@{}[ru]|(.7){=} & Y \ar[u]_{\eta_Y} \\ } &
\xymatrix@C=3pc@R=1pc{ \ar@{}[rd]|(.7){=} & M^2Y \\ 
   X \ar[r]_{\ff} \ar[ru]^{\kl{J(\ff)}} & MY \ar[u]_{\eta_{MY}} \\ } \\ 
\hline
K_M & 
\xymatrix@C=3pc{X \ar[r]^{f} & Y \\} &
\xymatrix@C=3pc{X\ar@{~>}[r]^{v=J(v_0)} & Y \\} &
\xymatrix@C=3pc@R=1pc{& MY \\ 
   X \ar[r]_{\lk{\ff}} \ar@{~>}[ru]^{J(\ff)} & Y \\ } \\  
\hline
\end{array} $$
In addition, the functor $J:C_0\to K_M$ has a right adjoint, 
which means that for each object $X$ there is an object $X^{\lift}$   
called the \emph{lifting} of $X$, 
with an isomorphism $K_M(X,Y) \cong C_0(X,Y^{\lift})$ natural in $X$ and $Y$.
Let us assume that the \emph{mono requirement} 
is satisfied by the monad, which means that $\eta_X$ is a mono for
every object $X$, or equivalently that the functor $J$ is faithful,
so that it defines an isomorphism from $C_0$ to $C_M$. 
\end{rem}

\subsection{Effects}
\label{subsec:effect}

In this section we define the effect of a morphism $f$
as a  kind of measure of how far $f$ is from being pure:
pure morphisms are effect-free and the effect of $v\circ f$, when $v$ is pure, 
is the same as the effect of $f$. 

\begin{defn}
\label{defn:effect}
Let $K$ be a category with a wide subcategory $C$
and with a pure terminal object $\uno$. 
\emph{The effect} of a morphism $f:X\to Y$ is the morphism 
$\cE(f)=\tu_Y\circ f:X\to\uno$.
We denote $f\seff f'$ when $f:X\to Y$ and $f':X\to Y'$ have the same effect:
  $$ \forall f:X\to Y \mysep \forall f':X\to Y' \mysep
  \formula{ f\seff f' \iff \tu_Y\circ f = \tu_{Y'}\circ f' } \;.$$
A morphism $f:X\to Y$ is \emph{effect-free} if $\cE(f)=\cE(\id_X)$,
which means that $\cE(f)=\tu_X$.
\end{defn}

The following properties are easily derived from the definition.

\begin{prop}
\label{prop:effect}
The same-effect relation $\seff$ 
is an equivalence relation between morphisms with the same domain that satisfies:
  \begin{itemize}
  \item \ppt{Pure morphisms are effect-free}.
  $\forall v:X\rpto Y \mysep  
  \formula{ v\seff \id_X } $.
  \item \ppt{Substitution}.
  $\forall f:X\to Y\mysep \forall g:Y\to Z \mysep \forall g':Y\to Z'\mysep
  \formula{ g\seff g' \implies g\circ f\seff g'\circ f } $.
  \item \ppt{Pure wiping}.  
  $\forall f:X\to Y \mysep \forall w:Y\rpto Z \mysep 
  \formula{ w\circ f\seff f } $.
  \end{itemize}
\end{prop}

\begin{rem}\textbf{Effects in a Kleisli category.} 
\label{rem:kl-effect}
Within the same framework as in remark~\ref{rem:kl-pure},
let us assume that there is a terminal object $\uno$  in $C_0$,
or equivalently in $C_M$. 
For each object $X$, the pure morphism $\tu_X:X\rpto\uno$ 
stands for $\kl{\tu_X}=\eta_\uno\circ \tu_X:X\to M\uno$ in $C_0$,
and for each morphism $f:X\to Y$ in $K_M$ 
the effect $\cE(f)$ of $f$ stands for 
$\kl{\tu_Y\circ f}=M\tu_Y \circ \kl{f}:X\to M\uno$ in $C_0$.
Let $\seff_0$ denote the relation between morphisms in $C_0$
defined by $\kl{f} \seff_0 \kl{f'}$ if and only if $f\seff f'$. Then in $C_0$: 
  $$ \forall \ff:X\to MY\mysep \forall \ff':X\to MY' \mysep 
     \formula{ \ff \seff_0 \ff' \iff M\tu_Y \circ \ff = M\tu_{Y'} \circ \ff' } \,.$$
\end{rem}

\subsection{Consistency} 
\label{subsec:value}

Now we define a consistency relation between two parallel morphisms.

\begin{defn}
\label{defn:purecons}
Let $K$ be a category with a wide subcategory $C$.
A \emph{consistency} relation $\cons$ is a relation between parallel morphisms, 
the second one being pure, which satisfies:
\begin{itemize}
\item \ppt{Pure reflexivity}.
  $ \forall v:X\rpto Y  \mysep
  \formula{ v\cons v } $.
  \item \ppt{Compatibility with composition}.   
  $ \forall f:X\to Y \mysep \forall g:Y\to Z \mysep 
  \forall u:Y\rpto Y' \mysep \forall v:X\rpto Y'  \mysep \forall w:Y'\rpto Z \mysep
  \\ \formula{ (u\circ f\cons v) \,\wedge\, (g \cons w \circ u) \implies g\circ f \cons w\circ v } $.
  $$ \xymatrix@C=3pc@R=1pc{
  X \ar@/_/[rd]_{f} \ar@{~>}[r]^{v}_{\consdown} & 
     Y' \ar@{~>}[r]^{w}_{\consdown} & Z \\ 
   & Y \ar@/_/[ru]_{g} \ar@{~>}[u]^(.4){u} & \\  } 
 \implies
  \xymatrix@C=3pc@R=1pc{ 
  X \ar@/^-2ex/[rr]_{g\circ f} \ar@{~>}@<1ex>[rr]^(.4){w\circ v}_{\consdown} & & Z \\ 
   }  $$
\end{itemize}
Two parallel morphisms $f$ and $f'$ are called \emph{consistent} 
when $f\cons v \rcons f'$ for some pure morphism $v$,
this is denoted $f\lrcons f'$. 
\end{defn}

The following properties are easily derived from the definition.

\begin{prop}
\label{prop:purecons}
Let $K$ be a category with a wide subcategory $C$
and with a consistency relation $\cons$.
Then: 
\begin{itemize} 
  \item \ppt{Preservation by composition}. 
  $ \forall f:X\to Y \mysep \forall v:X\rpto Y \mysep \forall g:Y\to Z  \mysep \forall w:Y\rpto Z \mysep
  \\ \formula{ (f \cons v) \wedge (g \cons w) \implies g\circ f \cons w\circ v } $.
  $$ \xymatrix@C=3pc{
  X \ar@/_3ex/[r]_{f} \ar@{~>}[r]^{v}_{\consdown} & 
    Y \ar@/_3ex/[r]_{g} \ar@{~>}[r]^{w}_{\consdown} & Z \\ 
  } \implies \xymatrix@C=3pc{
  X \ar@/_2ex/[rr]_{g\circ f} \ar@{~>}@<1ex>[rr]^(.4){w\circ v}_{\consdown} & & Z \\   }$$
  \item \ppt{Pure substitution}.  
  $ \forall v:X\rpto Y \mysep \forall  g:Y\to Z  \mysep \forall w:Y\rpto Z \mysep
  \formula{ g \cons w \implies g\circ v \cons w\circ v } $.
  \item \ppt{Pure replacement}.
  $ \forall f:X\to Y \mysep \forall v:X\rpto Y\mysep \forall w:Y\rpto Z \mysep
  \formula{ f\cons v  \implies w\circ f \cons w\circ v } $.
\end{itemize}
\end{prop}

\begin{defn}
\label{defn:value}
An \emph{effect category} $(C\subseteqq K,\cons)$
is made of a category $K$ and a wide subcategory $C$ of $K$, 
with a pure terminal object $\uno$ 
and the same-effect relation $\seff$ as in definition~\ref{defn:effect},
together with a consistency relation $\cons$ which satisfies:
\begin{itemize}
  \item \ppt{Complementarity with $\seff$}.
  $ \forall f,f':X\to Y \mysep 
  \formula{ (f\seff f') \,\wedge\, (f\lrcons f') \implies f=f' } $.
\end{itemize}
\end{defn}

In essence, the complementarity property can be stated as follows:
\slogan{if two morphisms have the same effect and are consistent, 
then they are equal.}

The following properties are easily derived.

\begin{prop}
\label{prop:value}
Let $(C\subseteqq K,\cons)$ be an effect category. 
Then: 
\begin{itemize} 
  \item \ppt{Consistency on effects}. 
  $ \forall f:X\to Y  \mysep
  \formula{ (\exists v \mysep f\cons v) \implies \cE(f) \cons \tu_X  }$.
  \item \ppt{Consistency on pure morphisms}. 
  $ \forall v,v':X\rpto Y  \mysep
  \formula{ v\cons v' \iff v=v' } $.
  \item \ppt{Consistency is unambiguous}.
  $ \forall f:X\to Y \mysep \forall v:X\rpto Y  \mysep
  \formula{ f\lrcons v \iff f\cons v } $.
\end{itemize}
\end{prop}

\begin{rem}
It follows that a pure morphism $v$ is consistent with itself 
and with no other pure morphism.
In general a morphism $f$ may be consistent 
with no pure morphism or with several ones.
The relation $\lrcons$ is symmetric but in general it is not reflexive.
\end{rem}

\begin{rem}
Let $K$ be a category with a wide subcategory $C$
and with a pure terminal object $\uno$. 
Then the same-effect relation $\seff$ is uniquely defined,
and there is a ``trivial'' consistency relation: the equality of pure morphisms. 
But neither the existence nor the unicity of a non-trivial 
consistency relation $\cons$ is guaranteed.
\end{rem}

\subsection{Extended consistency}
\label{subsec:effect-extended}

The consistency $\cons$ is a relation between two morphisms, 
the second one being pure.
It can be extended to pairs of arbitrary morphisms.

\begin{defn}
\label{defn:extended}
In an effect category $(C\subseteqq K, \cons)$, 
\emph{an extended consistency}  is a relation~$\Cons$ between parallel morphisms
such that:
\begin{itemize}
\item  \ppt{Extension}.  
  $ \forall f:X\to Y \mysep \forall v:X\rpto Y \mysep 
  \formula{ f \cons v \implies f\Cons v  } $. 
\item  \ppt{Substitution}. 
  $ \forall f:X\to Y \mysep \forall g,g':Y\to Z  \mysep
  \formula{ g \Cons g' \implies g\circ f \Cons g'\circ f } $.
\end{itemize}
The symmetric relation $\lrCons$ is defined by 
$f\lrCons f'$ if and only if there is a morphism $f''$ such that $f\Cons f'' \rCons f'$.
This relation $\lrCons$ is weaker than the relation $\lrcons$.
\end{defn}

It follows easily that $\Cons$ is reflexive and that $f \Cons f'$ implies $f\lrCons f'$. 

\begin{rem}
\label{rem:extended}
It is easy to check that 
in an effect category $(C\subseteqq K, \cons)$ 
there is a smallest extended consistency $\Cons$, which is defined as follows:
$\forall h,h':X\to Y \mysep$
  $$  h \Cons h' \iff \exists  f:X\to Y \mysep \exists g:Y\to Z \mysep 
  \exists w:Y\rpto Z \mysep 
  \formula{ (h=g\circ f)  \,\wedge\, (h'=w\circ f) \,\wedge\, (g \cons w )  }  $$
  $$ \xymatrix@C=3pc{
  X \ar[r]_{f} & Y \ar@/_3ex/[r]_{g} \ar@{~>}[r]^{w}_{\consdown} & Z \\ 
  }  \iff 
  \xymatrix@C=3pc{
  X \ar@/_2ex/[rr]_{g\circ f} \ar@{~>}@<1ex>[rr]^{w\circ f}_{\Consdown} & & Z \\ 
  }  $$
In addition, this relation $\Cons$ satisfies pure replacement:
\\  $ \forall f,f':X\to Y \mysep \forall w:Y\rpto Z  \mysep
  \formula{ f \Cons f' \implies w\circ f \Cons w\circ f' } $.
\end{rem}

\subsection{Examples of effect categories} 
\label{subsec:effect-exam}

Several examples are introduced in this section.
For each example, the same-effect relation $\seff$ is described,
then a consistency relation $\cons$ is chosen
in such a way that we get an effect category,
and the smallest extended consistency relation $\Cons$ is described.  
It will be checked in sections~\ref{subsec:cartesian-exam} 
and~\ref{subsec:monad-exam} that in each example 
the chosen consistency relation gives rise to a Cartesian effect category.
The examples about errors, lists, finite multisets and finite sets 
are provided directly by a monad $M$, then $K_M$ and $C_M$ 
are defined as in remark~\ref{rem:kl-pure}.
States could be treated with monads, at the cost of using an extra adjunction,
but this would not be possible for partiality over an arbitrary base category.

\exam{Errors} 
Let $C_0$ be a category 
with an initial object $\zero$ and 
with a distinguished object $E$ (for ``errors''),
hence with a unique morphism $\cotu_E:\zero\to E$.
Let us assume that there are coproducts of the form $X+E$ 
that \emph{behave well} in the sense of \emph{extensivity} \citep{CarLacWal93}: 
for every $\ff:X\to Y+E$, there is a coproduct $X=\cD{\ff}+\olcD{\ff}$ 
with two morphisms $\ff_Y:\cD{\ff}\to Y$ and $\ff_E:\olcD{\ff}\to E$ 
such that $\ff=\ff_Y+\ff_E$.
The \emph{error monad} on $C_0$ has $MX=X+E$ as endofunctor 
and the coprojection $\eta_X:X\to X+E$ as unit. 
A morphism $f:X\to Y$ in the Kleisli category $K_M$ stands for a morphism 
$\kl{f}:X\to Y+E$ in $C_0$, such that $\kl{f}=\kl{f}_Y+\kl{f}_E$ as explained above.
A pure morphism $v=J(v_0):X\rpto Y$ in $K_M$ stands for 
$\kl{v}=\eta_Y\circ v_0 : X\to Y+E$ in $C_0$, 
such that $\kl{v}=v_0+\cotu_E : X\to Y+E$ in $C_0$.
Let us assume that $C_0$ has a terminal object $\uno$.
For each morphism $f:X\to Y$ in $K_M$, 
the effect $\cE(f)=\tu_Y\circ f:X\to \uno$ is such that 
$\kl{\cE(f)} = (\tu_Y+\id_E)\circ \kl{f} = \tu_{\cD{\kl{f}}} + \kl{f}_E$. 
All this can be illustrated as follows in $C_0$, first for a pure morphism $v$
then for a morphism $f$ and finally for the effect $\cE(f)$; 
the vertical arrows are the coprojections:
$$ \xymatrix@C=3.5pc@R=1pc{
X \ar[r]^{v_0} \ar[d]_{\id_X} & Y\ar[d]  \\
X \ar[r]^{\kl{v}} \ar@{}[ru]|{=} \ar@{}[rd]|{=} & Y+E  \\
\zero \ar[r]^{\cotu_E} \ar[u]^{\cotu_X} & E \ar[u] \\
}\qquad
 \xymatrix@C=3.5pc@R=1pc{
\cD{\kl{f}} \ar[r]^{\kl{f}_Y} \ar[d] & Y \ar[d]  \\
X \ar[r]^{\kl{f}} \ar@{}[ru]|{=} \ar@{}[rd]|{=}  & 
  Y+E  \\
\olcD{\kl{f}} \ar[r]^{\kl{f}_E} \ar[u] & E \ar[u]  \\
}\qquad
 \xymatrix@C=3.5pc@R=1pc{
\cD{\kl{f}} \ar[r]^{\kl{f}_Y} \ar[d] \ar@/^4ex/[rr]^{\tu_{\cD{\kl{f}}}}_{=} & 
  Y \ar[r]^{\tu_Y} & \uno \ar[d] \\
X \ar[rr]^{\kl{\cE(f)}} \ar@{}[rru]|{=} \ar@{}[rrd]|{=} & 
   & \uno+E \\
\olcD{\kl{f}} \ar[r]^{\kl{f}_E} \ar[u] \ar@/_4ex/[rr]_{\kl{f}_E}^{=} & 
  E \ar[r]^{\id_E}  & E \ar[u] \\
}$$
Let $i_{\kl{f}}:\cD{\kl{f}} \to X$ denote the coprojection and 
let $\isoto$ denote an isomorphism in $C_0$.
\begin{itemize}
\item $ \forall f:X\to Y \mysep \forall f':X\to Y' \mysep 
  \formula{ f \seff f' \iff \exists i:\olcD{\kl{f}} \isoto \olcD{\kl{f'}} \mysep 
   \kl{f}_E = \kl{f'}_E \circ i } $. 
\item $ \forall f:X\to Y \mysep \forall v=J(v_0):X\rpto Y  \mysep 
  \formula{ f \cons v \iff  \kl{f}_Y = v_0 \circ i_{\kl{f}} } $.
\end{itemize}

When $C_0$ is the category of sets, 
we say that $\cD{\ff}$ is \emph{the domain of definition of $\ff$}
and that $\ff$ \emph{raises the error $e$ at $x$} whenever $\ff(x)=e\in E$,
so that a morphism $v$ is pure if and only if $\kl{v}$ does not raise any error.
Then, $f \seff f'$ means that $\kl{f}$ and $\kl{f'}$ 
raise the same errors for the same arguments,
hence they have the same domain of definition.
Furthermore, $f \cons v$ means that $\kl{f}$ coincides with $\kl{v}$ on $\cD{\kl{f}}$, 
hence $f \lrcons f'$ means that $\kl{f}$ and $\kl{f'}$ 
coincide on $\cD{\kl{f}} \cap \cD{\kl{f'}}$.
Then the smallest extended consistency relation is such that for all $f,f':X\to Y$, 
$f \Cons f'$ if and only if $\cD{\kl{f}} \subseteq \cD{\kl{f'}}$
and $\kl{f}$ coincides with $\kl{f'}$ on $\cD{\kl{f}}$ and also on $\olcD{\kl{f'}}$.
It follows that $\Cons$ is transitive
and that $\lrCons$ is the same relation as $\lrcons$. 

\exam{Partiality} 
A \emph{category of partial morphisms} is defined here, as in \citep{CurObt89},
as a category $K$ with a wide subcategory $C$ such that 
the category $K$ is enriched with a partial order $\leq$
and every pure arrow is maximal for $\leq$.
Then the morphisms in $K$ are called the \emph{partial functions} 
and the morphisms in $C$ the \emph{total functions}, 
as in the fundamental situation of sets.
In addition, let us assume that there is a pure terminal object $\uno$,
and wherefore the effect of a morphism $f:X\to Y$ is the morphism
$\tu_Y \circ f$ (in \citep{CurObt89} this morphism is called the \emph{domain of
  definition of $f$}). 
\begin{itemize}
\item $ \forall f:X\to Y \mysep \forall f':X\to Y' \mysep 
  \formula{ f \seff f'  \iff \tu_Y \circ f =\tu_Y \circ f' } $.
\item $ \forall f:X\to Y \mysep \forall v=J(v_0):X\rpto Y  \mysep 
  \formula{ f \cons v \iff f\leq v } $.
\item $ \forall f,f':X\to Y \mysep 
  \formula{ f \Cons f' \iff f\leq f' } $. 
\end{itemize} 
We add, as a new axiom, the complementarity of $\seff$ and $\cons$. 

On sets, with the usual notion of partial function, 
the inclusion of $C$ in $K$ has a right adjoint with lifting $X^{\lift}=X+\uno$, 
so that the partial functions from $X$ to $Y$
can be identified to the (total) functions from $X$ to $Y+\uno$
and the partial order $\leq$ corresponds to the inclusion 
of the domains of definition (in their usual sense, as subsets).
Then both points of view (partiality and error) are equivalent.

\exam{State}
Let $C_0$ be a category 
with a distinguished object $S$ (for ``states'')
and with products of the form $S\times X$. 
For each set $X$ let $\sigma_X:S\times X\to S$ and $\pi_X:S\times X\to X$ 
denote the projections.
Let $K$ be the category 
with the the same objects as $C_0$ and with a morphism $f:X\to Y$ 
for each $\kl{f}:S\times X\to S\times Y$ in $C_0$;
we say that $f$ in $K$ \emph{stands for} $\kl{f}$ in $C_0$.
Let $C$ be the wide subcategory of $K$ 
with the pure morphisms $v=J(v_0):X\rpto Y$ standing 
for $\kl{v}=\id_S\times v_0:S\times X\to S\times Y$.
Let us assume that $C_0$ has a terminal object $\uno$.
We may identify $S\times\uno$ with $S$, 
so that the morphism $\tu_X:X\rpto \uno$ stands for the projection $\sigma_X:S\times X\to S$
and the effect of a morphism $f:X\to Y$ stands for $\sigma_Y\circ \kl{f}:S\times X\to S$. 
\begin{itemize}
\item $ \forall f:X\to Y \mysep \forall f':X\to Y' \mysep 
  \formula{ f \seff f' \iff  \sigma_Y\circ \kl{f} = \sigma_{Y'}\circ \kl{f'} } $.
\item $ \forall f:X\to Y \mysep \forall v=J(v_0):X\rpto Y  \mysep 
  \formula{ f \cons v \iff \pi_Y\circ \kl{f} = v_0 \circ \pi_X } $.
\item $ \forall f,f':X\to Y  \mysep 
  \formula{ f \Cons f' \iff \pi_Y\circ \kl{f} = \pi_Y\circ \kl{f'} } $.
\end{itemize}
It follows that $\Cons$ is an equivalence relation, so that $\lrCons$ 
is the same as $\Cons$. 

On sets, $f \seff f'$ means that $\kl{f}$ and $\kl{f'}$ modify the state in the same way,
and $f \cons v$ means that $\kl{f}$ always returns the same value as $v_0$,
so that $f \lrcons f'$ means that $\kl{f}$ and $\kl{f'}$ 
both always return the same value, which in addition does not depend on the state,
while $f \Cons f'$ (as well as $f \lrCons f'$) means that $\kl{f}$ and $\kl{f'}$ 
both always return the same value, which may depend on the state.

\exam{Lists} 
Let us consider the \emph{list monad} with endofunctor $\cL$ 
on the category of sets.
The unit $\eta$ maps each $x$ to $(x)$ 
and the multiplication $\mu$ flattens each list of lists.
Since $\uno$ is a singleton, 
a list $\ell$ in $\cL(\uno)$ may be identified to its length $\len(\ell)$ in $\N$,
and the effect of a morphism $f:X\to Y$ to $\len\circ f:X\to \N$.
Then, a morphism $f$ is effect-free when $\len\circ f$ is the constant function~1.
For each $x\in X$ and $k\in\N$, we denote by $(x)^k$ the list 
$(x,\dots,x)$ where $x$ is repeated $k$ times. 
More generally, for each list $\ul{x}=(x_1,\dots,x_n)\in\cL(X)$ 
and each list of naturals $\ul{k}=(k_1,\dots,k_n)$ with the same length as $\ul{x}$,
we denote by $\ul{x}^{\ul{k}}$ the list 
$(x_1,\dots,x_1,\dots,x_n,\dots,x_n)$ where each $x_i$ is repeated $k_i$ times. 
\begin{itemize}
\item $ \forall f:X\to Y \mysep \forall f':X\to Y' \mysep 
  \formula{ f \seff f' \iff \forall x\in X \mysep \len(f(x)) =\len(f'(x)) }$.
\item $ \forall f:X\to Y \mysep \forall v=J(v_0):X\rpto Y  \mysep 
  \formula{ f \cons v \iff  \forall x\in X \mysep \exists k\in\N \mysep \kl{f}(x)=(v_0(x))^k\, } $. 
\item $ \forall f,f':X\to Y \mysep 
  \formula{ f \Cons f' \iff  \forall x\in X \mysep \exists \ul{k}\in\cL(\N) \mysep 
  \kl{f}(x)=\kl{f'}(x)^{\ul{k}} } $. 
\end{itemize}
It follows that $f \lrcons f'$ if and only if for each $x \in X$
there is some $y\in Y$ that is the unique element (if any) in the lists $\kl{f}(x)$ and $\kl{f'}(x)$,
and that $f \lrCons f'$ as soon as $f$ and $f'$ are parallel. 

\exam{Finite (multi)sets} 
The example of lists can easily be adapted to the \emph{finite multiset monad} 
and to the \emph{finite set monad} on the category of sets. 
For the finite multiset monad, $\cM(\uno)$ can be identified to $\N$
and the effect of a morphism to the cardinal of its image.
\begin{itemize}
\item $ \forall f:X\to Y \mysep \forall f':X\to Y' \mysep 
  \formula{ f \seff f' \iff \forall x\in X \mysep \card(f(x))=\card(f'(x)) } $.
\item $ \forall f:X\to Y \mysep \forall v=J(v_0):X\rpto Y  \mysep 
  \formula{ f \cons v \iff  
  \forall x\in X \mysep \kl{f}(x) \subseteq \{v_0(x)\} \, } $. 
\item $ \forall f,f':X\to Y \mysep 
  \formula{ f \Cons f' \iff  
  \forall x\in X \mysep \kl{f}(x) \subseteq \kl{f'}(x) \, } $. 
\end{itemize}
For the finite set monad, the definitions of $\cons$ and $\Cons$ are similar,
but $\seff$ is different. Since $\cP(\uno)$ has only
two elements $\emptyset$ and $\uno$,
we get $f \seff f'$ if and only if for all $x\in X$
either both $f(x)$ and $f'(x)$ are empty or both are non-empty.

\subsection{Results in evaluation logic} 
\label{subsec:result}

In \citep{Moggi95}, within the framework of \emph{evaluation logic} and
with respect to a strong monad satisfying some extra properties,
Moggi defines the relation $c\res a$, which means that  
the value $a$ is a \emph{result} of the computation $c$.
With the same notations as in remark~\ref{rem:kl-pure}, 
$c: \uno \to MX$ and $a: \uno \to X$ are morphisms in $C_0$, 
or equivalently $c=\kl{f}$ for a morphism $f:\uno\to X$ in $K_M$  
and $a=v_0:X\to Y$ yields a pure morphism $v=J(v_0):\uno\rpto X$.
Then it may happen that $f$ is \emph{consistent} with $v$ in the sense of this paper.
The following table compares both notions for several monads on sets. 
$$ \begin{array}{|c|c|c|}
\hline
\makebox[1.5cm]{Monad} & 
  \makebox[5.5cm]{Results \citep{Moggi95}} & 
  \makebox[5.5cm]{Consistency (this paper)} \\
MY    & c\res a  & f\cons v   \\
\hline
Y+E & 
c=a\; (\mbox{thus, c is total}) &
c\in Y \implies c=a \\
\hline
(Y\times S)^S &
\exists s\in S  \mysep \exists s'\in S \mysep c(s)=(a,s') &
\forall s\in S  \mysep \exists s'\in S \mysep c(s)=(a,s')  \\
\hline
\cL(Y) & 
a \in c &
\exists k\in\N  \mysep c = (a)^k \\
\hline
\cP(Y) &
a \in c &
c = \{a\} \mbox{ or } c = \emptyset \\  
\hline
\end{array}$$
From this table we see that
in general $f\cons v \not\Rightarrow c\res a$  and $c\res a  \not\Rightarrow  f\cons v$. 
It can easily be seen from the example of the state monad
that having the same results is not a consistency relation in general,
since two different morphisms may have the same effect and the same results.
Therefore, the notion of result in evaluation logic does not easily fit with 
our notion of consistency.

\section{Cartesian effect categories} 
\label{sec:cartesian}

\subsection{Cartesian categories}
\label{subsec:binary}

In this paper a \emph{Cartesian category} is a category with chosen finite products.
We denote by $\uno$ the terminal object, 
$\times$ for the products and $p,q,r,s,t,\dots$ (with indices) for the projections.
The binary product defines a functor $\times:C^2\to C$ such that 
for all $v_1:X_1\to Y_1$ and $v_2:X_2\to Y_2$, 
the morphism $v_1 \times v_2:X_1\times X_2 \to Y_1\times Y_2$ 
is the unique morphism that satisfies the \emph{binary product property}:
  $$ \begin{array}{cc}
   \xymatrix@R=.5pc{ \\ 
   \txt{ $q_1 \circ (v_1 \times v_2) = v_1\circ p_1$} \\ 
   \txt{ $q_2 \circ (v_1 \times v_2) = v_2\circ p_2$ } \\  }
  \qquad 
  & 
  \xymatrix@R=1.5pc@C=3pc{
  X_1 \ar[r]^{v_1} \ar@{}[rd]|{=} & Y_1 \\
  X_1\times X_2 \ar[r]^{v_1\times v_2} \ar[u]^{p_1} \ar[d]_{p_2} & 
  Y_1\times Y_2 \ar[u]_{q_1} \ar[d]^{q_2} \\
  X_2 \ar[r]^{v_2} & Y_2 \ar@{}[lu]|{=} \\
  } \\
  \end{array} $$

In a Cartesian category $C$, the \emph{swap} natural transformation $c$, 
with components $c_{X_1,X_2}:X_1\times X_2 \to X_2\times X_1$, 
is defined from the projections 
$p_i:X_1\times X_2 \to X_i$ and $p'_i:X_2\times X_1 \to X_i$ 
by $p'_i\circ c_{X_1,X_2} = p_i$ for $i=1,2$.
It follows that $c_{X_2,X_1}=c_{X_1,X_2}^{-1}$.

Now, Cartesian products in a category 
are generalized, first as semi-pure products, then as sequential products, 
in an effect category. 

\subsection{Semi-pure products} 
\label{subsec:semipure}

Let us consider an effect category $(C\subseteqq K,\cons)$ 
where $C$ is a Cartesian category. 
We define the \emph{semi-pure products} 
as two graph homomorphisms $\ltimessp:C\times K\to K$ 
and $\rtimessp:K\times C\to K$
that extend $\times$ and that 
satisfy some generalization of the binary product property 
involving the consistency relation $\cons$.
\slogan{while the universal property of a binary product consists in two equalities, 
the universal property of a semi-pure product consists in one equality and one consistency.}

\begin{defn}
\label{defn:semipure}
Let $(C\subseteqq K,\cons)$ be an effect category 
with a binary product $\times$ on $C$. 
A graph homomorphism $\ltimessp:C\times K\to K$
is \emph{the left semi-pure product} on $(C\subseteqq K,\cons,\times)$ 
if it extends $\times$ and satisfies the \emph{left semi-pure product property}:
for all $v_1:X_1\rpto Y_1$ and $f_2:X_2\to Y_2$, 
the morphism $v_1\ltimessp f_2:X_1\times X_2 \to Y_1\times Y_2$ 
is the unique morphism such that: 
  $$ \begin{array}{cc}
   \xymatrix@R=.5pc{ \\ 
   \txt{ $q_1\circ (v_1\ltimessp f_2) \cons v_1\circ p_1$} \\ 
   \txt{ $q_2\circ (v_1\ltimessp f_2) = f_2\circ p_2$ } \\  }
  \qquad 
  & 
\xymatrix@R=1.5pc@C=3pc{
X_1 \ar@{~>}[r]^{v_1} \ar@{}[rd]|{\consdown} & Y_1 \\
X_1\times X_2 \ar[r]^{v_1\ltimessp f_2} \ar@{~>}[u]^{p_1} \ar@{~>}[d]_{p_2} & 
  Y_1\times Y_2 \ar@{~>}[u]_{q_1} \ar@{~>}[d]^{q_2} \\
X_2 \ar[r]^{f_2} & Y_2 \ar@{}[lu]|{=} \\
} \\
  \end{array} $$
Symmetrically, a graph homomorphism $\rtimessp:K\times C\to K$
is \emph{the right semi-pure product} on $(C\subseteqq K,\cons,\times)$ 
if it extends $\times$ and satisfies the \emph{right semi-pure product property}:
for all $f_1:X_1\to Y_1$ and $v_2:X_2\rpto Y_2$,
the morphism $f_1\rtimessp v_2:X_1\times X_2 \to Y_1\times Y_2$ 
is the unique morphism such that: 
  $$ \begin{array}{cc}
   \xymatrix@R=.5pc{ \\ 
   \txt{ $q_1\circ (f_1\rtimessp v_2) = f_1\circ p_1$ } \\ 
   \txt{ $q_2\circ (f_1\rtimessp v_2) \cons v_2\circ p_2$ } \\  }
  \qquad 
  & 
 \xymatrix@R=1.5pc@C=3pc{
X_1 \ar@{}[rd]|{=} \ar[r]^{f_1} & Y_1 \\
X_1\times X_2 \ar[r]^{f_1\rtimessp v_2} \ar@{~>}[d]_{p_2} \ar@{~>}[u]^{p_1} & 
  Y_1\times Y_2 \ar@{~>}[u]_{q_1} \ar@{~>}[d]^{q_2} \\
X_2 \ar@{~>}[r]^{v_2} & Y_2  \ar@{}[lu]|{\consup}\\
} \\
  \end{array} $$
A \emph{Cartesian effect category}  is 
an effect category $(C\subseteqq K,\cons)$ 
with a binary product $\times$ on $C$ 
and with semi-pure products $\ltimessp$ and $\rtimessp$
(for short, it may be denoted $C\subseteqq K$ or simply $K$).
\end{defn}

A straightforward consequence of definition~\ref{defn:semipure} is that 
the right semi-pure product can be determined from the left one, as follows.
Consequently, from now on, we generally omit the right semi-pure products.

\begin{prop}
\label{prop:semipure-swap}
In a Cartesian effect category.
for all $f_1:X_1\to Y_1$ and  $v_2:X_2\rpto Y_2$: 
$$  (f_1\rtimes v_2) =  c_{Y_2,Y_1} \circ (v_2 \ltimes f_1) \circ c_{X_1,X_2} \;.$$
\end{prop}

In a binary product $v_1\times v_2$, obviously 
the first projection $q_1\circ (v_1\times v_2)$ does not depend on $v_2$, 
and symmetrically 
the second projection $q_2\circ (v_1\times v_2)$ does not depend on $v_1$.
For a left semi-pure product  $v_1\ltimes f_2$, 
this remains true for the second projection but not for the first one. 
However, a consequence of the complementarity of $\cons$ with $\seff$ 
is that $q_1\circ (v_1\ltimes f_2)$ depends on $f_2$ precisely through its effect $\cE(f_2)$,
as stated in the next proposition. 

\begin{prop}
\label{prop:semipure-complement}
In a Cartesian effect category,
for all $v_1:X_1\rpto Y_1$,  $f_2:X_2\to Y_2$ and $f'_2:X_2\to Y_2$,  
$ \cE(q_1\circ (v_1\ltimes f_2)) = \cE(v_1\ltimes f_2) = \cE(f_2\circ p_2)$ and:
$$ \cE(f_2)=\cE(f'_2) \implies q_1\circ (v_1\ltimes f_2) = q_1\circ (v_1\ltimes f'_2) \;.$$
\end{prop}

\begin{proof}
The first result derives from the pure wiping property of the effect. 
For the second result, let $h=v_1\ltimes f_2$ and $h'=v_1\ltimes f'_2$. 
The left semi-pure product property implies that 
$q_1\circ h \lrcons q_1\circ h'$ and $q_2\circ h = q_2\circ h'$.
The latter implies that $q_2\circ h \seff q_2\circ h'$, and thus by
pure wiping we have also $q_1\circ h \seff q_1\circ h'$.
The result now follows from  the complementarity of $\cons$ with $\seff$.
\end{proof}

The next proposition follows from the fact that the restriction of $\ltimessp$ to $C^2$ 
coincides with the binary product functor~$\times$ on $C$.
\begin{prop}
\label{prop:semipure-id}
In a Cartesian effect category,
for all objects $X_1$ and $X_2$: 
  $$ \id_{X_1}\ltimessp \id_{X_2}  = \id_{X_1}\times \id_{X_2} = \id_{X_1\times X_2} \;.$$
\end{prop}

\begin{rem}
\label{rem:semipure-unicity}
Let us assume that the following \emph{unicity condition} holds:
$$ \forall h,h':X\to Y_1\times Y_2 \mysep
\formula{ (q_1\circ h \lrcons q_1\circ h') \wedge (q_2\circ h = q_2\circ h')
\implies h=h' } \;. $$
In this case, if there is a graph homomorphism $\ltimessp:C\times K\to
K$ extending $\times$
and satisfying the left semi-pure product property, 
then $\ltimessp$ is the left semi-pure product.
\end{rem}

\subsection{Sequential products} 
\label{subsec:sequential}

In accordance with the intended meaning of ``sequential'', 
we define sequential products as composed from two
consecutive semi-pure products.

\begin{defn}
\label{defn:sequential}
In a Cartesian effect category,
the pair of \emph{sequential products} composed from 
the semi-products $\ltimessp$, $\rtimessp$
is made of the graph homomorphisms $\ltimesseq,\rtimesseq:K^2\to K$ 
(the \emph{left} and \emph{right} sequential products, respectively)
defined as follows:
\begin{itemize}
\item for all $f_1:X_1\to Y_1$ and $f_2:X_2\to Y_2$:
  $$ f_1\ltimesseq f_2 = (\id_{Y_1} \ltimessp f_2) \circ (f_1\rtimessp \id_{X_2}) $$
\item for all $f_1:X_1\to Y_1$ and $f_2:X_2\to Y_2$:
  $$ f_1\rtimesseq f_2 = (f_1\rtimessp \id_{Y_2}) \circ (\id_{X_1} \ltimessp f_2) $$
\end{itemize}
  $$ 
\xymatrix@R=2pc@C=2.5pc{
X_1 \ar@{}[rd]|{=} \ar[r]^{f_1} & Y_1  \ar@{~>}[r]^{\id} \ar@{}[rd]|{\consdown} & Y_1 \\
X_1\times X_2 \ar[r]^{f_1\rtimessp \id} \ar@{~>}[d]_{p_2} \ar@{~>}[u]^{p_1} & 
  Y_1\times X_2 \ar@{~>}[u]_{r_1} \ar@{~>}[d]^{r_2} \ar[r]^{\id\ltimessp f_2}  &
  Y_1\times Y_2 \ar@{~>}[u]_{q_1} \ar@{~>}[d]^{q_2} \\
X_2 \ar@{~>}[r]^{\id} & X_2  \ar@{}[lu]|{\consup} \ar[r]^{f_2} & Y_2 \ar@{}[lu]|{=}\\
} \qquad
\xymatrix@R=2pc@C=2.5pc{
X_1 \ar@{~>}[r]^{\id} \ar@{}[rd]|{\consdown} & X_1  \ar@{}[rd]|{=} \ar[r]^{f_1} & Y_1 \\
X_1\times X_2 \ar[r]^{\id\ltimessp f_2} \ar@{~>}[u]^{p_1} \ar@{~>}[d]_{p_2} & 
  X_1\times Y_2 \ar@{~>}[u]_{s_1} \ar@{~>}[d]^{s_2} \ar[r]^{f_1\rtimessp \id} & 
  Y_1\times Y_2 \ar@{~>}[u]_{q_1} \ar@{~>}[d]^{q_2} \\
X_2 \ar[r]^{f_2} & Y_2 \ar@{}[lu]|{=} \ar@{~>}[r]^{\id} & Y_2  \ar@{}[lu]|{\consup}\\
} $$
\end{defn}

It follows easily from proposition~\ref{prop:semipure-swap} that 
the right sequential product can be determined from the left one, as follows.
Consequently, from now on, we generally omit the right sequential products.

\begin{prop}
\label{prop:sequential-swap}
In a Cartesian effect category,
for all $f_1:X_1\to Y_1$ and  $f_2:X_2\to Y_2$: 
$$  (f_1\rtimesseq f_2) =  c_{Y_2,Y_1} \circ (f_2 \ltimesseq f_1) \circ c_{X_1,X_2} \;.$$
\end{prop}

\begin{prop}
\label{prop:sequential}
In a Cartesian effect category,
the left sequential product $\ltimesseq$ 
extends the left semi-pure product  $\ltimessp$. 
\end{prop}

\begin{proof}
Let $v:X_1\rpto Y_1$ and $f:X_2\to Y_2$. 
Since $v\ltimesseq f = (\id_{Y_1} \ltimessp f) \circ (v\rtimessp \id_{X_2}) $
and since $\rtimessp$ extends the binary product $\times$ on $C^2$:
  $$ v\ltimesseq f = (\id_{Y_1} \ltimessp f) \circ (v\times \id_{X_2}) \;.$$
The left semi-pure product property yields: 
  $$q_1\circ (\id_{Y_1} \ltimessp f) \cons r_1 
  \et 
  q_2\circ (\id_{Y_1} \ltimessp f) = f\circ r_2 $$
so that by pure substitution: 
  $$ q_1\circ (v\ltimesseq f) \cons r_1 \circ (v\times \id_{X_2})  
  \et 
  q_2\circ (v\ltimesseq f) = f\circ r_2  \circ (v\times \id_{X_2}) $$
hence from the binary product property we get: 
  $$ q_1\circ (v\ltimesseq f) \cons v\circ p_1 
  \et 
  q_2\circ (v\ltimesseq f) = f\circ p_2 $$
which is the left semi-pure product property.
\end{proof}

\begin{rem}
\label{rem:sequential}
It follows from proposition~\ref{prop:sequential}
that we may drop the subscript ``$\mathrm{seq}$''. 
\end{rem}

\begin{defn}
\label{defn:pairing}
In a Cartesian effect category,
for all $f_1:X\to Y_1$ and $f_2:X\to Y_2$ 
the \emph{left pairing} of $f_1$ and $f_2$ is 
$\tuple{f_1,f_2}_l=(f_1\ltimes f_2) \circ \tuple{\id_X,\id_X}:X\to Y_1\times Y_2$ 
and the \emph{right pairing} of $f_1$ and $f_2$ is 
$\tuple{f_1,f_2}_r=(f_1\rtimes f_2) \circ \tuple{\id_X,\id_X}:X\to Y_1\times Y_2$.
\end{defn}

\begin{rem}
\label{rem:sequential-direct-later}
Another point of view on sequential products, as ``direct'' generalizations of 
binary products (independently from any \textit{a priori} semi-pure products)  
is given in section~\ref{subsec:sequential-direct}.
\end{rem}

\subsection{Pure morphisms are central} 
\label{subsec:central} 

The next definition is similar to the 
definition of central morphisms in a binoidal category,  
see section~\ref{subsec:freyd}. 

\begin{defn}
\label{defn:central}
In a Cartesian effect category,
a morphism $k_1$ is \emph{central} if for each morphism $f_2$: 
 $$ k_1\ltimes f_2 = k_1\rtimes f_2 \;.$$ 
Then it follows from proposition~\ref{prop:sequential-swap} that 
$f_2\ltimes k_1 = f_2\rtimes k_1$.
The \emph{center} $C_K$ of $K$ is made of 
the objects of $K$ together with the central morphisms,
we will prove in theorem~\ref{thm:functor} that $C_K$ is a subcategory of $K$.
\end{defn}

\begin{rem}
\label{rem:central}
According to definition~\ref{defn:sequential}, in a Cartesian effect category 
a morphism $k_1:X_1\to Y_1$ is central if and only if for each morphism $f_2:X_2\to Y_2$:
  $$ (k_1\rtimes \id_{Y_2}) \circ (\id_{X_1} \ltimes f_2)
  = (\id_{Y_1} \ltimes f_2) \circ (k_1\rtimes \id_{X_2}) \;.$$
\end{rem}

\begin{rem}
\label{rem:central-id}
It follows from definition~\ref{defn:sequential} and proposition~\ref{prop:semipure-id} 
that the identities are central.
Theorem~\ref{thm:central} now proves that this is valid for all pure morphisms. 
\end{rem}

\begin{thm}
\label{thm:central}
In a Cartesian effect category, every pure morphism is central.
\end{thm}

\begin{proof}
Given $v:X_1\rpto Y_1$ and $f:X_2\to Y_2$,
let us prove that the left semi-pure product $v\ltimes f$
is equal to the right sequential product $v\rtimes f$. Let:
  $$ h = v\rtimes f = (v\rtimes \id_{Y_2}) \circ (\id_{X_1} \ltimes f)
   = (v\times \id_{Y_2}) \circ (\id_{X_1} \ltimes f) \;.$$
Using the binary product property:
 $$ q_1\circ h = v \circ s_1 \circ  (\id_{X_1} \ltimes f)
   \et  q_2\circ h = s_2 \circ  (\id_{X_1} \ltimes f) $$
then the left semi-pure product property:
  $$ s_1 \circ  (\id_{X_1} \ltimes f) \cons p_1 
   \et  s_2 \circ  (\id_{X_1} \ltimes f) = f \circ p_2 $$
we get by pure replacement: 
 $$ q_1\circ h  \cons v \circ p_1 
  \et  q_2\circ h = f \circ p_2 $$
which means that  the left semi-pure product property is satisfied: 
$h=v\ltimes f$, as required.
\end{proof}

\begin{rem}
\label{rem:times}
In view of theorem~\ref{thm:central}
there would be no ambiguity in denoting $\times$ for 
the semi-pure products $\ltimes$ and $\rtimes$, 
however we will not use this opportunity,
in order to keep in mind that the semi-pure products are not real products.
\end{rem}

\subsection{Functoriality properties} 
\label{subsec:functor}

As reminded in section~\ref{subsec:binary}, 
the binary product in a Cartesian category is a functor.
In this section it is proved that similarly the semi-pure products
in a Cartesian effect category are functors. 

\begin{lem}
\label{lem:functor-idcompose} 
In  a Cartesian effect category, 
for all $X_1$, $f_2:X_2\to Y_2$ and $g_2:Y_2\to Z_2$:  
$$  (\id_{X_1}\ltimes g_2) \circ (\id_{X_1}\ltimes f_2) = \id_{X_1} \ltimes (g_2\circ f_2)  \;.$$
\end{lem}

\begin{proof}
The proof is easily obtained by chasing the following diagram
and using the compatibility of consistency with composition.
$$ \xymatrix@R=1.5pc@C=2pc{
X_1 \ar@{~>}[r]^{\id} \ar@{}[rd]|{\consdown} & 
  X_1 \ar@{~>}[r]^{\id} \ar@{}[rd]|{\consdown} & X_1  \\
X_1\times X_2 \ar@{~>}[u]^{p_1} \ar@{~>}[d]_{p_2} \ar[r]^{\id\ltimes f_2} & 
  X_1\times Y_2 \ar@{~>}[u]_{s_1} \ar@{~>}[d]^{s_2} \ar[r]^{\id\ltimes g_2} &
  X_1\times Z_2 \ar@{~>}[u]_{s'_1} \ar@{~>}[d]^{s'_2}  \\
X_2 \ar[r]^{f_2} & 
  Y_2 \ar@{}[lu]|{=} \ar[r]^{g_2} & 
  Z_2 \ar@{}[lu]|{=}  \\
}$$
\end{proof}

\begin{lem}
\label{lem:functor-compose} 
In  a Cartesian effect category,
for all $f_1:X_1\to Y_1$, $k_1:Y_1\to Z_1$, $f_2:X_2\to Y_2$ and $g_2:Y_2\to Z_2$
with $k_1$ central:  
$$  (k_1\ltimes g_2) \circ (f_1\ltimes f_2) = (k_1\circ f_1) \ltimes (g_2\circ f_2) $$
\end{lem}

\begin{proof}
According to definition~\ref{defn:sequential}:
$$  (k_1\ltimes g_2) \circ (f_1\ltimes f_2) = 
  (\id_{Z_1} \ltimes g_2) \circ (k_1\rtimes \id_{Y_2})  \circ
 (\id_{Y_1} \ltimes f_2) \circ (f_1\rtimes \id_{X_2})  \;. $$
Since $k_1$ is central, this is equal to 
$ (\id_{Z_1} \ltimes g_2) \circ  (\id_{Z_1} \ltimes f_2) \circ 
 (k_1\rtimes \id_{X_2})  \circ (f_1\rtimes \id_{X_2})$.
The result now follows from lemma~\ref{lem:functor-idcompose} 
and definition~\ref{defn:sequential} again.
\end{proof}

\begin{thm}
\label{thm:functor} 
In a Cartesian effect category $C\subseteqq K$, 
the center $C_K$ is a wide subcategory of $K$ that contains~$C$,
and the restrictions of the sequential products 
are functors $\ltimes:C_K\times K\to K$ and $\rtimes:K\times C_K\to K$.
\end{thm}

\begin{proof}
The central morphisms form a subcategory of $K$:
this comes from remark~\ref{rem:central-id} for identities 
and from lemma~\ref{lem:functor-compose}  and its symmetric version 
for composition.
The center $C_K$ is wide by definition, 
and it contains $C$ because of theorem~\ref{thm:central}.
The restrictions of the left sequential product is a functor: 
by proposition~\ref{prop:semipure-id} for identities 
and lemma~\ref{lem:functor-compose}  for composition.
Symmetrically, the restrictions of the right sequential product is a functor.
\end{proof}

\subsection{Naturality properties} 
\label{subsec:natural}

As reminded in section~\ref{subsec:binary}, 
a Cartesian category $C$ with $\times:C^2\to C$ and $\uno$ 
forms a symmetric monoidal category,
which means that the projections can be combined
in order to get natural isomorphisms $a,r,l,c$ with components:
\begin{itemize}
\item $a_X=a_{X_1,X_2,X_3}:(X_1\times X_2)\times X_3 \to X_1\times (X_2\times X_3)$, 
\item $r_X:\uno\times X\to X$, $l_X:X\times \uno\to X$, 
\item $c_X=c_{X_1,X_2}:X_1\times X_2 \to X_2\times X_1$, 
\end{itemize}
which satisfy the symmetric monoidal coherence conditions \citep{Maclane97}.
In this section we prove that in a Cartesian effect category $C\subseteqq K$, 
the natural isomorphisms $a,r,l,c$ that are defined from $C$
satisfy more general naturality conditions,  
involving the sequential products $\ltimes,\rtimes$.
The verification of the next result is straightforward from the definitions. 

\begin{lem}
\label{lem:monoidal} 
In a Cartesian effect category,
for all $f_1$, $f_2$, $f_3$ and pure $v_1$, $v_2$, $v_3$: 
$$ \left\{ \begin{array}{l}
   a_Y \circ (f_1\rtimes (v_2\rtimes v_3)) = ((f_1 \rtimes v_2)\rtimes v_3) \circ a_X \\
   a_Y \circ (v_1\ltimes (f_2\rtimes v_3)) = ((v_1 \ltimes f_2)\rtimes v_3) \circ a_X \\
   a_Y \circ (v_1\ltimes (v_2\ltimes f_3)) = ((v_1 \ltimes v_2)\ltimes f_3) \circ a_X \\
  \end{array} \right. $$
\end{lem}

\begin{thm}
\label{thm:monoidal} 
In a Cartesian effect category,
for all $f:X\to Y$, $f_1:X_1\to Y_1$, $f_2:X_2\to Y_2$ and $f_3:X_3\to Y_3$: 
$$ \left\{ \begin{array}{l}
  r_Y \circ (\id_{\uno} \ltimes f) = f \circ r_X \\ 
  l_Y \circ (f \rtimes \id_{\uno})  = f \circ l_X \\
  c_Y \circ (f_1\rtimes f_2) =  (f_2 \ltimes f_1) \circ c_X \\ 
  a_Y \circ (f_1\ltimes (f_2\ltimes f_3)) = ((f_1 \ltimes f_2)\ltimes f_3) \circ a_X \\
  a_Y \circ (f_1\rtimes (f_2\rtimes f_3)) = ((f_1 \rtimes f_2)\rtimes f_3) \circ a_X \\ 
  \end{array} \right. $$
\end{thm}

\begin{proof}
Since $r_X$ and $l_X$ are the projections, 
the first two lines comes from the definition of semi-pure products.
Since $c_X$ is the swap morphism from section~\ref{subsec:binary},
the third line is proposition~\ref{prop:sequential-swap}. 
As for the fourth line, let us use the definition of sequential products:
$$ f_1\ltimes (f_2\ltimes f_3) =
(\id \ltimes  (f_2\ltimes f_3)) \circ (f_1 \rtimes \id) 
\et  f_2\ltimes f_3 = (\id \ltimes f_3) \circ (f_2 \rtimes \id)$$
hence by lemma~\ref{lem:functor-idcompose}: 
$$ \id \ltimes  (f_2\ltimes f_3) = 
  (\id \ltimes  (\id \ltimes f_3)) \circ (\id \ltimes (f_2 \rtimes \id)) $$
and finally:
$$ f_1\ltimes (f_2\ltimes f_3) =
    (\id \ltimes  (\id \ltimes f_3)) \circ (\id \ltimes (f_2 \rtimes \id)) \circ (f_1 \rtimes \id) \;.$$
In a symmetric way:
$$ (f_1\ltimes f_2)\ltimes f_3 =
  (\id \ltimes  f_3) \circ ((\id \ltimes f_2) \rtimes \id) \circ ((f_1 \rtimes \id) \rtimes \id) \;. $$
Hence the result follows from the three lines of lemma~\ref{lem:monoidal}, 
together with proposition~\ref{prop:semipure-id} for dealing with identities.
\end{proof}

\subsection{The sequential product properties} 
\label{subsec:sequential-direct}

Sequential products also satisfy the 
\emph{left and right sequential product properties},
as defined below,
which generalize the binary product property. 
We use an extended consistency $\Cons$, as defined in 
section~\ref{subsec:effect-extended}.

\begin{defn}
\label{defn:sequential-direct}
Let $(C\subseteqq K,\cons)$ be an effect category
with an extended consistency relation $\Cons$ and with a pair of
graph homomorphisms $\ltimes',\rtimes':K^2\to K$ extending $\times$.
Then
the \emph{left sequential product property} states that 
for all $f_1:X_1\to Y_1$ and $f_2:X_2\to Y_2$, 
the morphism $f_1\ltimes' f_2:X_1\times X_2 \to Y_1\times Y_2$ 
satisfies:
  $$ \begin{array}{cc}
   \xymatrix@R=.5pc{ \\  \\ 
   \txt{ $q_1\circ (f_1\ltimes' f_2) \Cons f_1\circ p_1$} \\ 
   \txt{ $q_2\circ (f_1\ltimes' f_2) = f_2\circ r_2\circ (f_1\rtimes' \id_{X_2})$ } \\  }
  \quad 
  & 
\xymatrix@R=1pc@C=2.5pc{
X_1 \ar[rr]^{f_1} \ar@{}[rrdd]|{\Consdown} && Y_1 \\
&& \\
X_1\times X_2 \ar[rr]^{f_1\ltimes' f_2} \ar@{~>}[uu]^{p_1} \ar@/_3ex/[rd]_{f_1\rtimes' \id} && 
  Y_1\times Y_2 \ar@{~>}[uu]_{q_1} \ar@{~>}[dd]^{q_2} \\
& Y_1\times X_2 \ar@{~>}[d]^{r_2} & \\
& X_2 \ar[r]^{f_2} & Y_2 \ar@{}[luu]|(.6){=} \\
  } \\
  \end{array} $$
Symmetrically, 
the \emph{right sequential product property} says that   
for all $f_1:X_1\to Y_1$ and $f_2:X_2\to Y_2$, 
the morphism $f_1\rtimes' f_2:X_1\times X_2 \to Y_1\times Y_2$ 
satisfies:
  $$ \begin{array}{cc}
   \xymatrix@R=.5pc{ \\  \\ 
   \txt{ $q_1\circ (f_1\rtimes' f_2) = f_1\circ s_1\circ (\id_{X_1} \ltimes' f_2)$} \\ 
   \txt{ $q_2\circ (f_1\rtimes' f_2) \Cons f_2\circ p_2$ } \\  }
  \quad 
  & 
  \xymatrix@R=1pc@C=2.5pc{
& X_1 \ar@{}[rdd]|(.4){=} \ar[r]^{f_1} & Y_1 \\
& X_1\times Y_2 \ar@{~>}[u]_{s_1} & \\
X_1\times X_2 \ar[rr]^{f_1\rtimes' f_2} \ar@{~>}[dd]_{p_2} \ar@/^3ex/[ru]^{\id \ltimes' f_2} && 
  Y_1\times Y_2 \ar@{~>}[uu]_{q_1} \ar@{~>}[dd]^{q_2} \\
&& \\
X_2 \ar[rr]^{f_2} & & Y_2  \ar@{}[lluu]|{\Consup}\\
  } \\
  \end{array} $$
\end{defn}

\begin{prop}
\label{prop:sequential-direct}
In a Cartesian effect category,
the sequential products $\ltimes,\rtimes$ satisfy the sequential product properties.
\end{prop}

\begin{proof}
The left sequential product is defined as 
$ f_1\ltimes f_2 = (\id_{Y_1} \ltimessp f_2) \circ (f_1\rtimessp \id_{X_2}) $.
Since $\Cons$ extends $\cons$, 
the left semi-pure product property yields: 
  $$q_1\circ (\id_{Y_1} \ltimessp f_2) \Cons r_1 
  \et 
  q_2\circ (\id_{Y_1} \ltimessp f_2) = f_2\circ r_2 $$
so that by the substitution property of $\Cons$: 
  $$ q_1\circ (f_1\ltimes f_2) \Cons r_1 \circ (f_1\rtimessp \id_{X_2})  
  \et 
  q_2\circ (f_1\ltimes f_2) = f_2 \circ r_2  \circ (f_1\rtimessp \id_{X_2}) \;.$$
The right semi-pure product  property implies that 
$ r_1\circ (f_1\rtimessp \id_{X_2}) = f_1\circ p_1 $,
hence:
  $$ q_1\circ (f_1\ltimes f_2) \Cons f_1\circ p_1 
  \et 
  q_2\circ (f_1\ltimes f_2) = f_2 \circ r_2  \circ (f_1\rtimes \id_{X_2}) $$
which is the left sequential product property.
\end{proof}

\begin{rem}
\label{rem:sequential-direct-unicity}
The following condition is called the \emph{extended unicity} condition: 
$$ \forall h,h':X\to Y_1\times Y_2 \mysep
\formula{ (q_1\circ h \lrCons q_1\circ h') \,\wedge\, (q_2\circ h = q_2\circ h')
\implies h=h' } $$
Since $\lrCons$ is weaker than $\lrcons$, 
the extended unicity condition implies the unicity condition 
of remark~\ref{rem:semipure-unicity}. 
Whenever the extended unicity condition holds, 
the sequential product properties can be used as a definition 
of the sequential products, instead of definition~\ref{defn:sequential}.
In addition, although this looks like a mutually recursive definition of the  
left and right sequential products, this recursivity has only two steps.

Indeed, let $\ltimes,\rtimes$ be the sequential products
and let $f_1:X_1\to Y_1$ and $f_2:X_2\to Y_2$.
First let $h=f_1\rtimes \id_{X_2}$. 
The right semi-pure product property states that 
$q_1\circ h = f_1\circ p_1$ and $q_2\circ h \cons v_2\circ p_2$,
thanks to the unicity condition this is a characterization of $h$.
Now let $k=f_1\ltimes f_2$,
from proposition~\ref{prop:sequential-direct} we get
$q_1\circ k \Cons f_1\circ p_1$ and $q_2\circ k = f_2\circ r_2\circ h$,
and thanks to the extended unicity condition this is a characterization of $k$.
\end{rem}

\subsection{Some examples of Cartesian effect categories} 
\label{subsec:cartesian-exam}

In this section and in section~\ref{subsec:monad-exam}
we check that the effect categories 
from section~\ref{subsec:effect-exam} can be seen as Cartesian effect categories.
In each example, for any pure morphism $v$ and morphism $f$
we build a morphism $v\ltimes f$, and it is left as an exercise to check that 
$v\ltimes f$ actually is the left semi-pure product of $v$ and $f$.
In addition, it happens that the extended unicity condition is satisfied, 
so that the sequential products are characterized by the sequential product properties.

\exam{Errors} 
According to \citep{CarLacWal93}, an extensive category with products 
is  distributive. So, in the category $C_0$,
for all $X$, $Y$, $Z$ 
the canonical map from $X\times Y + X\times Z$ to $X\times (Y+Z)$ is an isomorphism. 
Let $v=J(v_0):X_1\rpto Y_1$ and $f:X_2\to Y_2$ in $K$,
so that by distributivity $X_1\times X_2$ is isomorphic to 
$(X_1\times \cD{\kl{f}})+(X_1\times \olcD{\kl{f}})$.
We define $v\ltimes f:X_1\times X_2 \to Y_1\times Y_2$ by 
$\cD{\kl{v\ltimes f}}=X_1\times \cD{\kl{f}}$, 
$\olcD{\kl{v\ltimes f}}=X_1\times \olcD{\kl{f}}$, 
$\kl{v\ltimes f}_Y=v_0\times \kl{f}_Y$ and $\kl{v\ltimes f}_E=\kl{f}_E \circ \pi$, 
where $\pi:X_1\times \olcD{\kl{f}}\to\olcD{\kl{f}}$ is the projection.
$$ \xymatrix@C=4pc@R=1pc{ 
  \cD{\kl{f}} \ar[r]^{\kl{f}_Y} \ar[d]  &  Y_2 \ar[d]  \\
  X_2 \ar[r]^{\kl{f}} \ar@{}[rd]|{=} \ar@{}[ru]|{=} & 
    Y_2 + E  \\
  \olcD{\kl{f}} \ar[r]^{\kl{f}_E} \ar[u] & E \ar[u] \\
  } \qquad\qquad 
  \xymatrix@C=4pc@R=1pc{ 
  X_1\times \cD{\kl{f}} \ar[r]^{v_0\times \kl{f}_Y} \ar[d]  &  Y_1\times Y_2 \ar[d]  \\
  X_1\times X_2 \ar[r]^{\kl{v\ltimes f}} \ar@{}[rd]|{=} \ar@{}[ru]|{=} & 
    Y_1\times Y_2 + E  \\
  X_1\times \olcD{\kl{f}} \ar[r]^{\kl{f}_E \circ \pi} \ar[u] & E \ar[u] \\
  } $$

On sets, as expected, this provides the left sequential product: 
$ \forall x_1 \in X_1 \mysep \forall x_2 \in X_2 \mysep $
$$ 
(f_1\ltimes f_2) (x_1,x_2) =  
\left\{\begin{array}{ll}
\tuple{\kl{f_1}(x_1),\kl{f_2}(x_2)} & \mbox{ if } \kl{f_1}(x_1) \in Y_1 \et  \kl{f_2}(x_2) \in Y_2 \\
\kl{f_2}(x_2) &  \mbox{ if } \kl{f_1}(x_1) \in Y_1 \et  \kl{f_2}(x_2) \in E \\ 
\kl{f_1}(x_1) & \mbox{ if } \kl{f_1}(x_1) \in E \\
\end{array}\right. $$
When $E$ has one element all morphisms are central, but
as soon as $E$ has more than one element there are non-central morphisms.

\exam{Partiality} 
Given a category of partial morphisms, 
if we impose the existence of sequential products and the fact that all morphisms are central, 
then we get a notion that is rather similar to  
the notion of \emph{partial Cartesian} category of partial morphisms in \citep{CurObt89}. 

On sets, up to adjunction, the left sequential product is the same as for the monad $X+1$: 
  $ \cD{(f_1\ltimes f_2)} = \cD{f_1} \ltimes \cD{f_2}$ and  
  $$ \forall x_1 \in \cD{f_1} \mysep \forall x_2 \in \cD{f_2} \mysep
  (f_1\ltimes f_2) (x_1,x_2) = \tuple{\kl{f_1}(x_1),\kl{f_2}(x_2)} \;.$$

\exam{State}
Let $v=J(v_0):X_1\rpto Y_1$ and $f:X_2\to Y_2$ in $K$.
Let us define $v\ltimes f:X_1\times X_2 \to Y_1\times Y_2$,
up to the relevant commutations, by 
$\kl{v\ltimes f} = v_0 \times \kl{f} : S\times X_1\times Y_1 \to S\times Y_1\times Y_2$.
$$ \xymatrix@C=4pc@R=1pc{ 
  X_1 \ar[r]^{v_0}  &  Y_1  \\
  S\times X_1\times X_2 \ar[r]^{\kl{v\ltimes f}}  \ar[u] \ar[d] \ar@{}[rd]|{=} \ar@{}[ru]|{=} & 
    S\times Y_1\times Y_2 \ar[u] \ar[d] \\
  S\times X_2 \ar[r]^{\kl{f}} & S\times Y_2 \\
  } $$

On sets, as expected, this provides the left sequential product: 
$$ \forall x_1 \in X_1 \mysep \forall x_2 \in X_2 \mysep \forall s\in S  \mysep 
  \kl{f_1\ltimes f_2}  (s,x_1,x_2) = \tuple{s_2,y_1,y_2} $$
where $\kl{f_1}(s,x_1) = \tuple{s_1,y_1}$ and $\kl{f_2}(s_1,x_2) = \tuple{s_2,y_2}$.
The left sequential product $f_1\ltimes f_2$ 
is usually distinct from the right sequential product $f_1\rtimes f_2$.

\section{Comparisons}
\label{sec:compare}

The use of \emph{strong monads} for dealing with computational effects 
has been introduced by Moggi for reasoning about programs
\citep{Moggi89,Moggi91,Wadler92}. 
This has been generalized by Power and Robinson,
who defined \emph{Freyd-categories} and proved that a strong monad 
is equivalent to a Freyd-category with an adjunction 
\citep{PowerRobinson97,PowerThielecke99}.
Independently, Arrows have been introduced by Hughes for generalizing 
strong monads in Haskell \citep{Hughes00, Paterson01};
it was believed that Arrows are ``essentially'' equivalent to Freyd-categories,
until Atkey proved that Arrows are in fact more general than Freyd categories \citep{Atkey08}.
In this section we directly compare each of these three frameworks 
to Cartesian effect categories:
Freyd-categories in section~\ref{subsec:freyd},
Arrows in section~\ref{subsec:arrows} and
strong monads in section~\ref{subsec:monads}.
Examples are considered in section~\ref{subsec:monad-exam}.

\subsection{Freyd-categories}
\label{subsec:freyd}

In this section, it is proved that Cartesian effect categories are 
Freyd-categories \citep{PowerRobinson97,PowerThielecke99,Selinger01}.
Let $|K|$ denote the smallest wide subcategory of $K$,
made of the objects and identities of $K$.

\begin{defn} 
\label{defn:binoidal}
A \emph{binoidal category} is a category $K$ together with  
two functors $\otimes:|K|\times K \to K$ and $\otimes:K\times |K| \to K$ 
which coincide on $|K|^2$
(so that the notation $\otimes$ is not ambiguous). 
The functors $\otimes$ can be extended as two graph homomorphisms 
$\ltimesfr,\rtimesfr:K^2\to K$, as follows. 
For all $f_1:X_1\to Y_1$ and $f_2:X_2\to Y_2$ in $K$, let:
$$ \left\{ \begin{array}{l} 
  f_1\ltimesfr f_2 = (\id_{Y_1} \otimes f_2) \circ (f_1\otimes\id_{X_2}) : 
 X_1\otimes X_2\to Y_1\otimes Y_2 \\ 
  f_1\rtimesfr f_2 = (f_1\otimes \id_{Y_2}) \circ (\id_{X_1} \otimes f_2) : 
 X_1\otimes X_2\to Y_1\otimes Y_2 \\
  \end{array} \right. $$
A morphism $k_1:X_1\to Y_1$ is \emph{central} if for all $f_2:X_2\to Y_2$, 
$k_1\ltimesfr f_2 = k_1\rtimesfr f_2$ and symmetrically $f_2\ltimesfr k_1 = f_2\rtimesfr k_1$.
Let $t:\Phi\To\Psi$ be a natural transformation between two functors $\Phi,\Psi:K'\to K$, 
then $t$ is \emph{central} if every component of $t$ is central.
\end{defn}
In theorem~\ref{thm:freyd} the graph homomorphisms $\ltimesfr,\rtimesfr$ 
will be related to the sequential products $\ltimes,\rtimes$ from section~\ref{sec:cartesian}.
In the next definition, ``natural'' means natural in each component separately. 

\begin{defn}
\label{defn:premonoidal}
A \emph{symmetric premonoidal category} is a binoidal category $K$ together with 
an object $I$ of $K$ and central natural isomorphisms with components
$a_{X,Y,Z}:(X\otimes Y)\otimes Z \to X\otimes (Y\otimes Z)$, 
$l_X:X\otimes I\to X$, 
$r_X:I\otimes X\to X$  and 
$c_{X,Y}:X\otimes Y \to X\otimes Y$,
subject to the usual coherence equations for symmetric monoidal categories 
\citep{Maclane97}.
Note that every symmetric monoidal category,
hence every category with finite products, is symmetric premonoidal.
A \emph{symmetric premonoidal functor} between two 
symmetric premonoidal categories is a functor that preserves 
the partial functor $\otimes$, the object $I$ and the natural isomorphisms $a,l,r,c$. 
It is \emph{strict} if in addition it maps central morphisms to central morphisms.
A \emph{Freyd-category} is an identity-on-objects functor $J:C\to K$ where 
the category $C$ has finite products, 
the category $K$ is symmetric premonoidal 
and the functor $J$ is strict symmetric premonoidal. 
\end{defn}

The following result states that every Cartesian effect category is a Freyd-category.
It is an easy consequence of the results in section~\ref{sec:cartesian}.

\begin{thm}
\label{thm:freyd}
Let $C\subseteqq K$ be a Cartesian effect category.
Let $a,l,r,c$ be the natural isomorphisms on $C$ defined as in section~\ref{subsec:natural}.
Let $J:C\to K$ be the inclusion, let $\otimes:|K|\times K \to K$ and $\otimes:K\times|K| \to K$
be the restrictions of $\ltimes$ and $\rtimes$, respectively, and let $I=\uno$.
This forms a Freyd-category, where $\ltimesfr$ and $\rtimesfr$
coincide with $\ltimes$ and $\rtimes$, respectively.
\end{thm}

\begin{proof}
The graph homomorphisms $\otimes:|K|\times K \to K$ and $\otimes:K\times |K| \to K$
coincide on $|K|^2$, and they are functors by theorem~\ref{thm:functor}, 
hence $K$ with $\otimes$ is a binoidal category.
Then, definitions~\ref{defn:sequential} and~\ref{defn:binoidal} 
state that the graph homomorphisms $\ltimesfr,\rtimesfr$ 
are the sequential products $\ltimes,\rtimes$.
It follows that both notions of central morphism
(definitions~\ref{defn:central} and~\ref{defn:binoidal}) coincide. 
The fact that the transformations $a,l,r,c$ are natural, 
in the sense of symmetric premonoidal categories, 
is an immediate consequence of theorem~\ref{thm:monoidal}
(in fact for $a$ it is lemma~\ref{lem:monoidal}).
Since all the components of $a,l,r,c$ are defined from the symmetric monoidal category $C$,
we know that they are isomorphisms and that they satisfy the coherence equations. 
In addition, since all pure morphisms are central by theorem~\ref{thm:central}, 
it follows that $a,l,r,c$ are central.
Hence $K$ with $\otimes$, $I$ and $a,l,r,c$ is a symmetric premonoidal category.
Clearly the inclusion functor $J:C\to K$ is symmetric premonoidal, 
and it is strict because of theorem~\ref{thm:central}.
\end{proof}

\subsection{Arrows}
\label{subsec:arrows}
\renewcommand{\arraystretch}{1} 

In view of the similarities between Freyd-categories and Arrows, 
it can be guessed that every Cartesian effect category 
gives rise to an Arrow \citep{Hughes00, Paterson01}; 
this is stated in this section.

\begin{defn}
\label{defn:related-arr}
An \emph{Arrow type} is a binary type constructor $\tA$ of the form:
\\ $\begin{array}{l}
\quad \texttt{class Arrow $\tA$ where} \\
\quad \quad \arr :: (X\to Y)\to \tA\;X\;Y \\
\quad \quad (\acomp) :: \tA\;X\;Y \to \tA\;Y\;Z \to \tA\;X\;Z \\
\quad \quad \first :: \tA\;X\;Y \to \tA\;(X,Z)\;(Y,Z) \\
\end{array}$
\\ satisfying the following equations: 
\begin{center}
\begin{tabular}{crcl}
 (1) & $\arr\; \id \acomp f $ &=& $ f$ \\ 
 (2) & $f \acomp \arr\; \id $ &=& $ f$ \\ 
 (3) & $(f \acomp g) \acomp h $ &=& $ f \acomp (g \acomp h)$ \\ 
 (4) & $\arr\;(w.v) $ &=& $ \arr\; v \acomp \arr\; w$ \\ 
 (5) & $\first\; (\arr\; v) $ &=& $ \arr\; (v\times\id)$ \\
 (6) & $\first\;(f \acomp g) $ &=& $ \first\; f \acomp \first\; g$ \\ 
 (7) & $\first\; f \acomp \arr\; (\id\times v) $ &=& 
    $ \arr\; (\id\times v) \acomp \first\; f$ \\
 (8) & $\first\; f \acomp \arr\; \fst $ &=& $ \arr\; \fst \acomp f$ \\
 (9) & $\;\;\first\; (\first\; f) \acomp \arr\; \assoc $ &=& 
    $ \arr\; \assoc \acomp \first\; f$ \\
\end{tabular}
\end{center}
where the functions $(\times)$, $\fst$ and $\assoc$ are defined as: 
\\ $\begin{array}{l}
(\times) :: (X\to X')\to(Y\to Y')\to (X,Y)\to (X',Y') \,\mbox{ such that }\,
(f\times g)(x,y)=(f\,x,g\,y)  \\ 
\fst :: (X,Y)\to X \;\mbox{ such that }\;
\fst(x,y)=x \\
\assoc :: ((X,Y),Z)\to (X,(Y,Z)) \;\mbox{ such that }\;
\assoc((x,y),z) = (x,(y,z)) \\
\end{array}$
\end{defn}

Let $C_H$ denote the category of Haskell types and ordinary functions, 
so that the Haskell notation $\mathtt{(X\to Y)}$ represents $C_H(X,Y)$, 
made of the Haskell ordinary functions from $X$ to $Y$. 
An arrow $\tA$ constructs a type $\tA\;X\;Y$ for all types $X$ and $Y$.
We slightly modify the definition of Arrows 
by allowing $\mathtt{(X\to Y)}$ to represent $C(X,Y)$ for any 
Cartesian category $C$ 
and by requiring that $\tA\;X\;Y$ is a set rather than a type:
more on this issue can be found in \citep{Atkey08}.
In addition, we use categorical notations instead of Haskell syntax. 
For this reason, from now on, for any Cartesian category $C$, 
an \emph{Arrow $A$ on $ C$} associates to each objects 
$X$, $Y$ of $ C$ a set $A(X,Y)$, together with three operations: 
$\arr :  C(X,Y)\to A(X,Y) \mysep
\acomp :  A(X,Y) \to A(Y,Z) \to A(X,Z)  \mysep 
\first :  A(X,Y) \to A(X\times Z,Y\times Z)  \mysep$
that satisfy the equations (1)--(9).
Basically, the correspondence between a Cartesian effect category $C\subseteqq K$
and an Arrow $A$ on $ C$ 
identifies $ K(X,Y)$ with $A(X,Y)$ for all types $X$ and $Y$.
This is stated more precisely in proposition~\ref{prop:arrows}.

\begin{prop}
\label{prop:arrows}
Every Cartesian effect category $C\subseteqq K$ 
gives rise to an Arrow $A$ on $C$, according to the following table:
$$ \begin{array}{|l|l|}
\hline
\boxEff & \boxArr \\
\hline
K(X,Y) &
  A(X,Y) \\ 
C(X,Y)\subseteq K(X,Y) &
  \arr:  C(X,Y)\to A(X,Y) \\ 
f\mapsto (g\mapsto g\circ f) & 
  \acomp: A(X,Y) \to A(Y,Z) \to A(X,Z) \\
f\mapsto f\times\id & 
  \first: A(X,Y) \to A(X\times Z,Y\times Z)  \\ 
\hline
\end{array} $$
\end{prop}

\begin{proof}
The first and second line in the table say that $A(X,Y)$ is 
made of the morphisms from $X$ to $Y$ in $ K$
and that $\arr$ is the conversion from pure morphisms to arbitrary morphisms. 
The third and fourth lines say that 
$\acomp$ is the (reverse) composition of morphisms 
and that $\first$ is the right semi-pure product with the identity. 
Now we prove that $A$ is an Arrow 
by translating each property (1)--(9) in terms of Cartesian effect categories 
and giving the argument for its proof.
Note that $\fst$ is the common name for projections like $p_1,q_1,\dots$ 
(in section~\ref{sec:cartesian}) 
and that $\assoc$ is the natural isomorphism $a$ as in section~\ref{subsec:natural}.
\begin{center}
\begin{tabular}{crcll}
 (1) & $f\circ \id $ &=& $f$       & 
  identity in $ K$ \\ 
 (2) & $ \id \circ f $ &=& $f$    & 
  identity in $ K$ \\ 
 (3) & $h\circ (g\circ f) $ &=& $ (h\circ g)\circ f $ & 
  associativity in $ K$ \\ 
 (4) & $w\circ v$ in $ C$ &=& $w\circ v$ in $ K$ &
  $C\subseteq K$ is a functor \\ 
 (5) & $v\times\id $ in $ C$ &=& $ v\times\id$ in $ K$ & 
  $\times$ in $K$ extends $\times$ in $C$ \\   
 (6) & $(g\circ f)\times\id $ &=& $ (g\times\id)\circ (f\times\id)$ &
  lemma~\ref{lem:functor-idcompose} \\
 (7) & $(\id\times v) \circ (f\times\id) $ &=& $ (f\times\id) \circ (\id\times v)$ & 
  theorem~\ref{thm:central} \\
 (8) & $q_1\circ (f\times\id) $ &=& $ f\circ p_1$ &
  definition~\ref{defn:semipure} \\
 (9) & $a\circ ((f\times\id)\times\id) $ &=& $ (f\times\id)\circ a$ & 
  lemma~\ref{lem:monoidal} \\
\end{tabular}
\end{center}
\end{proof}

The Arrow combinators $\scond$, $(\altimes)$ and $(\atuple)$ 
can be derived from $\arr$, $(\acomp)$ and $\first$, see e.g \citep{Hughes00,Paterson01}.
The correspondence in proposition~\ref{prop:arrows} is easily extended
to these functions. 
The left pairing $\tuple{f_1,f_2}_l$ 
and the natural isomorphism $c$ (corresponding to $\swap$)
are defined in section~\ref{subsec:sequential} and~\ref{subsec:natural},
respectively.
\begin{center}
\begin{tabular}{|l|l|}
\hline
\boxEff & \boxArr \\
\hline
$ (id \times f) =  c \circ (f \times id)  \circ  c  $ &
  $\scond\;f  = \arr\;\swap \acomp \first\;f \acomp \arr\;\swap $ \\
$f_1 \ltimes f_2  =  (\id\times f_2) \circ (f_1\times\id)$ &
  $f_1 \altimes f_2 = \first\;f_1 \acomp \scond\;f_2 $ \\ 
$\tuple{f_1,f_2}_l = (f_1\ltimes f_2) \circ \tuple{\id,\id}$ &
  $f_1 \atuple f_2  = \arr(\lambda x \rightarrow (x,x)) \acomp (f_1 \altimes f_2) $ \\
\hline
\end{tabular} 
\end{center}

For instance in \citep[\S 4.1]{Hughes00} it is stated that 
$\atuple$ is not a categorical product since in general $f_1$
is different from $(f_1 \atuple f_2) \acomp \arr\;\fst$:
``\emph{there is no reason to expect Haskell's pair type, $\atuple$, to be a
categorical product in the category of arrows, or indeed to expect any
categorical product to exist}''.
We can state this more precisely in a Cartesian effect category,
where $(f_1 \atuple f_2) \acomp \arr\;\fst$ corresponds to $q_1 \circ \tuple{f_1,f_2}_l$. 
Indeed, both morphisms are consistent: it follows from 
proposition~\ref{prop:sequential-direct} and pure substitution that
$ q_1 \circ \tuple{f_1,f_2}_l  \Cons  f_1$.

\subsection{Strong monads}
\label{subsec:monads}

Strong monads correspond to Freyd-categories $J:C\to K$ with a right adjoint for $J$
\citep{PowerRobinson97},
while Cartesian effect categories correspond to Freyd-categories 
with a sequential product (theorem~\ref{thm:freyd}). 
In this section, we give a condition which characterizes 
the strong monads such that the corresponding Freyd-category 
is a \emph{weak} Cartesian effect category, 
which means that there are two 
graph homomorphisms  $\ltimessp:C\times K\to K$ and $\rtimessp:K\times C\to K$
which satisfy the left and right semi-pure product property respectively, 
but which may not be unique.

We use the same notations as in remark~\ref{rem:kl-pure}.
It has been seen in remark~\ref{rem:kl-effect}
that the effect of a morphism $f:X\to Y$ of $K$ 
stands for $\kl{\tu_Y\circ f}=M\tu_Y \circ \kl{f}: X \to M\uno$ in $C_0$,  
so that in $C_0$: 
  $$ \forall \ff:X\to MY \mysep \forall \ff':X\to MY' \mysep
     \formula{ \ff \seff_0 \ff' \iff M\tu_Y\circ \ff = M\tu_{Y'}\circ \ff' } \,.$$
Let $\cons$ be a consistency relation on $C\subseteqq K$,
then the relation $\cons_0$ in $C_0$ is defined by
$\kl{f} \cons_0 \kl{v} \iff f\cons v$, or equivalently:
  $$ \forall \ff,\ff':X\to MY \mbox{ in }C_0 \mysep 
  \formula{ \ff \cons_0 \ff' \iff 
  \exists v_0:X\to Y  \mbox{ in }C_0 \mysep (\ff'=\eta_Y\circ v_0) \,\wedge\,
  (\lk{\ff} \cons J(v_0) )} \,.$$ 
The pure substitution property of $\cons$ (proposition~\ref{prop:value}) 
corresponds to the following \ppt{substitution property of~$\cons_0$}: 
  $$ \forall v_0:X\to Y \mysep \forall w_0:Y\to Z  \mysep \forall \gg :Y\to MZ \mysep
  \formula{ \gg \cons_0 \eta_Z\circ w_0 \implies
  \gg \circ v_0 \cons_0 \eta_Z\circ w_0 \circ v_0}\,.$$

Now in addition let us assume that $C_0$, hence $C$, is Cartesian. 
In \citep{Moggi89},
it is explained why the monad $(M,\mu,\eta)$ and the product $\times$
are not sufficient for dealing with several variables:
there is a \emph{type mismatch} from $Y_1\times MY_2$ to $M(Y_1\times Y_2)$. 
This issue is solved by adding a \emph{strength},
i.e., a natural transformation $t$ with components 
$t_{Y_1,Y_2}:Y_1\stimes MY_2 \to M(Y_1\stimes Y_2)$
satisfying four axioms \citep{Moggi89}. 
One of these axioms is that for all $X$, 
$r_{MX} = Mr_X\circ t_{\uno,X}:\uno\stimes MX\to MX$, 
where the natural isomorphism $r$ is made of the projections
$r_X:\uno\stimes X\to X$ as in section~\ref{subsec:natural}.
Let us assume that we are given a strength $t$ for our monad.
In $K$, let $v:X_1\rpto Y_1$ and $f:X_2\to Y_2$;
in order to form a kind of product of $v$ and $f$, 
the usual method 
consists in composing in $C_0$ 
the product $v_0\times\kl{f}:X_1\stimes X_2 \to Y_1\stimes MY_2$
with the strength $t_{Y_1,Y_2}:Y_1\stimes MY_2 \to M(Y_1\stimes Y_2)$; 
we call this construction the \emph{left Kleisli product}.
The right Kleisli product is defined symmetrically.

\begin{defn}
\label{defn:monads-prod}
For all $v=J(v_0):X_1\rpto Y_1$ and $f:X_2\to Y_2$ in $K$, 
the \emph{left Kleisli product} of $v$ and $f$ in $K$ is defined by: 
  $$ \kl{v\ltimeskl f} = t_{Y_1,Y_2} \circ (v_0\times\kl{f}) : 
  X_1\times X_2 \to M(Y_1\times Y_2) \mbox{ in } C_0  \;.$$ 
\end{defn}

\begin{lem}
\label{lem:strength}
The strength can be expressed as a left Kleisli product:
  $$  \lk{t_{Y_1,Y_2}} = \id_{Y_1} \ltimeskl \,\lk{\id_{MY_2}} \mbox{ in } K  \;.$$
For all $Y_1,Y_2$, with projections $q_2:Y_1\times Y_2\rpto Y_2$ 
and $q'_2:Y_1\times MY_2 \rpto MY_2$: 
  $$ q_2 \circ \lk{t_{Y_1,Y_2}} = q'_2  \mbox{ in } K \;.$$
\end{lem}

\begin{proof}
In $K$, let $v=\id_{Y_1}:Y_1\rpto Y_1$ and $f=\lk{\id_{MY_2}}:MY_2\to Y_2$,
so that $v_0=\id_{Y_1}$ and $\kl{f}=\id_{MY_2}$ in~$C_0$.
Then $v_0\times\kl{f} = \id_{Y_1\times MY_2}$
so that $\kl{v\ltimeskl f} = t_{Y_1,Y_2}$, this is the first property. 
Now, for readability, we omit the subscript $0$ for naming the projections in $C_0$. 
The result is equivalent to $Mq_2 \circ t_{Y_1,Y_2} = q'_2$ in $C_0$.
The projection $q_2$ can be decomposed as 
$q_2=r_2\circ (\tu_{Y_1}\times Y_2) $, 
where $r_2=r_{Y_2}:\uno\times Y_2\to Y_2$ is the projection. 
Hence on the one hand $Mq_2=Mr_2 \circ M(\tu_{Y_1}\times Y_2) $,
and on the other hand $q'_2=r'_2\circ (\tu_{Y_1}\times MY_2) $
where $r'_2=r_{MY_2}:\uno \times MY_2\to MY_2$ is the projection. 
$$ \xymatrix@R=1.5pc{
  Y_1\times MY_2 \ar@/_12ex/[dd]_{q'_2}^{\quad =} \ar[d]_(.6){\tu\times M\id} 
    \ar[r]^{t_{Y_1,Y_2}} \ar@{}[rd]|{=} & 
    M(Y_1\times Y_2) \ar[d]^(.6){M(\tu\times\id)}  \ar@/^12ex/[dd]^{Mq_2}_{=\quad } \\
  \uno \times MY_2 \ar[d]_{r'_2} \ar[r]^{t_{\uno,Y_2}} \ar@{}[rd]|{=} & 
    M(\uno\times Y_2) \ar[d]^{Mr_2} \\
  MY_2 \ar[r]^{\id} & MY_2 \\
}$$
In the previous diagram, the square on the top is commutative 
since $t$ is natural, and the square on the bottom is commutative 
thanks to the property of the strength with respect to $r$. 
Hence the large square is commutative, and the result follows.
\end{proof}

\begin{thm}
\label{thm:monads}
Let $C_0$ be a Cartesian category with a strong monad $(M,\mu,\eta,t)$
and with a consistency relation $\cons$ on $C\subseteqq K$. 
Then $C_0$ with the left and right Kleisli products 
is a weak Cartesian effect category if and only if 
for all $Y_1,Y_2$ (with the projections $q_1:Y_1\times Y_2\to Y_1$ 
and $q'_1:Y_1\times MY_2 \to Y_1$):
  $$ q_1 \circ \lk{t_{Y_1,Y_2}} \cons q'_1 \mbox{ in } K \mysep 
  \;\mbox{ or equivalently }\; 
  Mq_1 \circ t_{Y_1,Y_2} \cons_0 \eta_{Y_1} \circ q'_1 \mbox{ in } C_0  \;.$$
If in addition $\forall \ff,\ff':X\to M(Y_1\times Y_2) $ in $C_0\mysep$
$$ 
\formula{ (Mq_1\circ \ff \lrcons_0 Mq_1\circ \ff') \wedge (Mq_2\circ \ff = Mq_2\circ \ff')
\implies \ff=\ff' } \mbox{ in } C_0  \;, $$
then $C_0$ with the left and right Kleisli products 
is a Cartesian effect category.
\end{thm}

Roughly speaking (i.e., forgetting the projections), this means that
$C_0$ with the Kleisli products is a weak Cartesian effect category if
and only if:
\slogan{the strength of the monad is consistent with the identity.} 

\begin{proof}
Let us consider the morphism $\lk{t_{Y_1,Y_2}}$.
By the first part of lemma~\ref{lem:strength} 
$\lk{t_{Y_1,Y_2}} = \id_{Y_1} \ltimeskl \,\lk{\id_{MY_2}}$.
Therefore, if the left Kleisli product does satisfy the left semi-pure product property,
then $q_1 \circ \lk{t_{Y_1,Y_2}} \cons q'_1$. 
Now, let us assume that $q_1 \circ \lk{t_{Y_1,Y_2}} \cons q'_1$;
this is illustrated below, together with $q_2 \circ \lk{t_{Y_1,Y_2}} = q'_2$ 
(second part of lemma~\ref{lem:strength}), first in $K$ then in $C_0$:
$$ \xymatrix@C=4pc@R=1.5pc{ 
  Y_1 \ar@{~>}[r]^{\id}  & Y_1 \\
  Y_1\times MY_2 \ar@{~>}[u]^{q'_1} \ar@{~>}[d]_{q'_2} \ar[r]^{\lk{t}} 
    \ar@{}[rd]|{=} \ar@{}[ru]|{\consdown} & 
    Y_1\times Y_2 \ar@{~>}[u]_{q_1} \ar@{~>}[d]^{q_2} \\
  MY_2 \ar[r]^{\lk{\id}} & Y_2 \\
  }  
\qquad 
 \xymatrix@C=4pc@R=1.5pc{ 
  Y_1 \ar[r]^{\eta}  & MY_1 \\
  Y_1\times MY_2 \ar[u]^{q'_1} \ar[d]_{q'_2} \ar[r]^{t} 
    \ar@{}[rd]|{=} \ar@{}[ru]|{\consdown_0} & 
    M(Y_1\times Y_2) \ar[u]_{Mq_1} \ar[d]^{Mq_2} \\
  MY_2 \ar[r]^{\id} & MY_2 \\
  }  
$$
For any $v:X_1\rpto Y_1$ and $f:X_2\to Y_2$, the morphism $v\ltimeskl f$ in $K$ 
is defined by $\kl{v\ltimeskl f} = t_{Y_1,Y_2} \circ (v_0\times\kl{f})$ in $C_0$. 
In the diagram below, in $C_0$, 
the left-hand side illustrates the binary product property of $v_0\times\kl{f}$
and the right-hand side is as above. 
$$\xymatrix@C=4pc@R=1.5pc{ 
  X_1 \ar[r]^{v_0} \ar@/^4ex/[rr]^{\kl{v}}_{=} & Y_1 \ar[r]^{\eta}  & MY_1 \\
  X_1\times X_2 \ar[u]^{p_1} \ar[d]_{p_2} \ar[r]^{v_0\times\kl{f}}
    \ar@{}[rd]|{=} \ar@{}[ru]|{=} & 
    Y_1\times MY_2 \ar[u]^{q'_1} \ar[d]_{q'_2} \ar[r]^{t} 
      \ar@{}[rd]|{=} \ar@{}[ru]|{\consdown_0} & 
    M(Y_1\times Y_2) \ar[u]_{Mq_1} \ar[d]^{Mq_2} \\
  X_2 \ar[r]^{\kl{f}} \ar@/_4ex/[rr]_{\kl{f}}^{=} & MY_2 \ar[r]^{\id} & MY_2 \\
  }  
$$
It follows immediately from the bottom part of this diagram that  
$Mq_2 \circ \kl{v\ltimeskl f} = \kl{f}\circ p_2$,
which means that
$ q_2\circ (v\ltimeskl f) = f\circ p_2 \mbox{ in } K$.
Moreover, it follows from the top part, using the substitution property of $\cons_0$, that 
$Mq_1 \circ \kl{v\ltimeskl f} \cons_0 \kl{v} \circ p_1$,
which means that
$ q_1\circ (v\ltimeskl f) \cons v\circ p_1  \mbox{ in } K $.
The left semi-pure product property is hence satisfied by $\ltimeskl$.

Then the last part of the theorem follows immediately 
from remark~\ref{rem:semipure-unicity}.
\end{proof}

\subsection{More examples of Cartesian effect categories} 
\label{subsec:monad-exam}

In this section we consider the effect categories in section~\ref{subsec:effect-exam} 
which are defined from a strong monad.
In each example the strength is described, 
then it is easy to check that the conditions of theorem~\ref{thm:monads}
are satisfied, so that the Kleisli category gives rise to
a cartesian effect category with the Kleisli products as semi-pure products. 
However, for the monads of lists and of finite (multi)sets,
the extended consistency relation is so weak that the 
sequential product properties (definition~\ref{defn:sequential-direct})
are not sufficient for characterizing the sequential products.

\exam{Errors} 
\label{ex:monad-error}
The strength $t_{X_1,X_2}$ is obtained by composing the isomorphism 
$X_1\times(X_2+E) \cong (X_1\times X_2)+(X_1\times E)$
with $\id_{X_1\times X_2} + \sigma_{X_1}: 
(X_1\times X_2)+(X_1\times E) \to (X_1\times X_2)+E$, 
where $\sigma_{X_1}$ is the projection.
The Kleisli products are semi-pure products from section~\ref{subsec:cartesian-exam}. 

\exam{Lists} 
The strength is such that 
for all $x_1\in X_1$ and $\ulx_2=(x_{2,1},\dots,x_{2,k}) \in\cL(X_2)$,
$\formula{ t_{X_1,X_2}(x_1,\ulx_2) = (\tuple{x_1,x_{2,1}},\dots,\tuple{x_1,x_{2,k}} ) }$.
It follows that $Mp_1 \circ t_{X_1,X_2}(x_1,\ulx_2) = (x_1)^k$ 
while $\eta_{X_1} \circ p'_1(x_1,\ulx_2) = (x_1)$.
So, the left sequential product is:
$$ \forall x_1\in X_1 \mysep \forall x_2\in X_2 \mysep 
\formula{ (f_1\ltimes f_2) (x_1,x_2) =  
(\tuple{y_1,z_1},\dots,\tuple{y_1,z_p},\dots,\tuple{y_n,z_1},\dots,\tuple{y_n,z_p}) } \;,$$
where $f_1(x_1)=(y_1,\dots,y_n)$ and $f_2(x_2)=(z_1,\dots,z_p)$, 
so that there are non-central morphisms. 

\exam{Finite (multi)sets} 
Finite multisets and finite sets have similar properties. 
For sets, the strength is such that 
for all $x_1\in X_1$ and $\ulx_2 \in\cP(X_2)$,  
$\formula{ t_{X_1,X_2}(x_1,\ulx_2) = \{ \tuple{x_1,x'} \mid x'\in \ulx_2 \} }$,
and both the left and the right sequential product are:
$$ \forall x_1\in X_1 \!\mysep\! \forall x_2\in X_2 \!\mysep\! 
\formula{  (f_1\!\ltimes\! f_2) (x_1,x_2) \!=\!   (f_1\!\rtimes\! f_2) (x_1,x_2) \!=\!  
\{ \tuple{y,z} \mid y\in f_1(x_1) \wedge z\in f_2(x_2) \}} \,.$$

\section{Conclusion}

This paper deals with the major issue of 
formalizing computational effects, 
especially while using multivariate functions. 
For this purpose, we have introduced several new features: 
first a \emph{consistency} relation
and the associated notion of \emph{effect category}, 
then the  \emph{semi-pure} and \emph{sequential} products 
for getting a \emph{Cartesian effect category}.
Thanks to the universal property of the semi-pure products,
each Cartesian effect category is endowed with 
a powerful tool for definitions and proofs. 
This has been used for proving that every Cartesian effect category 
is a Freyd-category 
and for giving conditions which 
ensure that a strong monad gives rise to a Cartesian effect category.
We have studied several examples of effects, 
in each case we get a Cartesian effect category.


Since the notions of effect category and Cartesian effect category are new,
there is still a large amount of work to do in order to study 
their applications and their limitations.
For instance, in order to define some kind of closure, 
one could try to generalize the results of \citep{CurObt89} on
partiality to other effects. 
Further investigations include:
enhancing the comparison with \citep{Moggi95} in order to clarify the
relations between Cartesian effect categories and evaluation logic;
fitting more examples in our framework (e.g. continuations).
In addition, the issue of combining effects, as in \citep{HylandEtal06}, 
might be revisited from the point of view of effect categories.


\subsubsection*{Acknowledgments}
The authors would like to thank Eugenio Moggi for pointing out 
the papers \citep{CurObt89} and \citep{Moggi95}.

\end{document}